\documentclass{CSML}

\def\dOi{12(2:9)2016}
\lmcsheading%
{\dOi}
{1--27}
{}
{}
{May\phantom.~\phantom02, 2014}
{Jun.~29, 2016}
{}

\ACMCCS{[{\bf Theory of computation}]: Formal languages and automata
  theory; Semantics and reasoning---Program reasoning---Program
  analysis}

\usepackage{amsmath}
\usepackage{amssymb}
\usepackage{url}
\usepackage{hyperref}

\usepackage{tikz}
\usetikzlibrary{automata}
\usetikzlibrary{arrows,decorations.pathmorphing,backgrounds,positioning,fit}

\makeatletter
\newcommand{\xRightarrow}[2][]{\ext@arrow 0359\Rightarrowfill@{#1}{#2}}
\makeatother
\newcommand{\set}[1]{\{#1\}}

\newcommand{\tup}[1]{\langle{#1}\rangle}

\newcommand{\conf}[2]{\langle #1,#2 \rangle}
\newcommand{\ptrans}[5]{\conf{#1}{#2} \stackrel{#5}{\hookrightarrow} \conf{#3}{#4}}

\newcommand{\utransrule}[4]{#1 \xRightarrow[#4]{#3} #2}

\newcommand{\atrans}[4]{#1 \xrightarrow{#3 \, \mid \, #4} #2}
\newcommand{\atransx}[5]{#1 \xrightarrow[#5]{#3 \, \mid \, #4} #2}

\newcommand{\atransrule}[4]{#1 \xRightarrow{#3 \, \mid \, #4} #2}

\newcommand{\atransrulex}[5]{#1 \underset{#5}{\overset{#3 \, \mid \, #4}{\Longrightarrow}} #2}

\newenvironment{restatement}[1]{\noindent{}{\textbf{#1}.}\itshape}{}

\newcommand{\ext}[2]{\uparrow_{#1}\!(#2)}
\newcommand{\extop}[1]{\uparrow_{#1}}

\newcommand{\al}{\#}
\newcommand{\idealof}[1]{\mathcal{I}(#1)}
\newcommand{\upcl}[1]{{#1}^{\uparrow}}

\newcommand{\scomp}[2]{{#1}\parallel{#2}}

\newcommand{\sheight}[1]{\mathrm{height}(#1)}

\newcommand{\PFUN}{\mathrm{PFun}}
\renewcommand{\emptyset}{\varnothing}
\newcommand{\trstep}[3]{#1_{#2,#3}}

\makeatletter
\newcommand\mytup[1]{%
  \@tempcnta=0
  \langle
  \@for\@ii:=#1\do{%
    \@insertbreakingcomma
    \@ii
  }%
  \rangle
}
\def\@insertbreakingcomma{%
  \ifnum \@tempcnta = 0 \else\,,\linebreak[1] \fi
  \advance\@tempcnta\@ne
}
\makeatother

\title[Weighted Pushdown Systems with Indexed Weight Domains]{Weighted Pushdown Systems with Indexed Weight Domains\rsuper*}
\author[Y.~Minamide]{Yasuhiko Minamide}
\address{Department of Mathematical and Computing Sciences, Tokyo Institute of Technology, Japan
}
\email{minamide@is.titech.ac.jp}  

\keywords{pushdown system, reachability analysis, semiring}
\titlecomment{{\lsuper*}A preliminary version of this article appeared in the proceedings of the 19th International Conference on Tools and Algorithms for the Construction and Analysis of Systems (TACAS), LNCS 7795, pp.~230--244, 2013.}

\begin{document}
\maketitle

\begin{abstract}
The reachability analysis of weighted pushdown systems is a very powerful technique in verification and analysis of 
recursive programs. Each transition rule of a weighted pushdown system is associated with an element of a bounded semiring representing the weight of the rule. However, we have realized that the restriction of the boundedness is too strict and the formulation of weighted pushdown systems is not general enough for some applications.

To generalize weighted pushdown systems, we first introduce the notion of stack signatures that summarize the effect of a computation of a pushdown system and formulate pushdown systems as automata over the monoid of stack signatures. We then generalize weighted pushdown systems by introducing semirings indexed by the monoid and weaken the boundedness to local boundedness.
\end{abstract}

\section{Introduction}
The reachability analysis of weighted pushdown systems is a very powerful
technique in verification and analysis of recursive programs~\cite{reps:weighted}. 
Each transition rule of a weighted pushdown system is associated with
an element of a semiring representing the weight of the rule.
To guarantee termination of the analysis, the semiring of the weight
must be bounded: there should be no infinite descending sequence of weights.
However, recently, we have realized that this restriction of the boundedness is
too strict and the formulation of weighted pushdown systems is not 
general enough for some applications. 
For the two applications below, the standard algorithm for the reachability
analysis of weighted pushdown systems actually works and terminates.
However, they require semirings that are not bounded and thus the standard
framework of weighted pushdown systems cannot guarantee termination.

The first application is the reachability analysis of conditional pushdown
systems. 
Conditional pushdown systems extend 
pushdown systems 
with the ability to check the whole stack content against a regular 
language~\cite{esparza:regularvaluation,li:conditional}. 
We proposed an algorithm of their reachability analysis 
in our previous work on the analysis of the HTML 5 parser specification~\cite{minamide:html5}.
After the development of the algorithm, we realized that the algorithm 
can be considered as the reachability analysis of weighted pushdown 
systems. However, it required an unbounded semiring.

The second application is the analysis of recursive programs with local
variables. For the efficient analysis of
recursive programs, Suwimonteerabuth proposed an encoding of 
local variables into weight implemented with BDDs~\cite{SuwimonteerabuthPhd}.
The weight has a structure depending on a configuration of stack 
and requires a semiring that is not bounded.

\smallskip

To generalize weighted pushdown systems, we first introduce \emph{stack signatures} that summarize the effect of a computation of a pushdown system 
as a pair of words over a stack alphabet. 
A stack signature $w_1/w_2$ represents a computation of a pushdown system
that pops $w_1$ and pushes $w_2$ as its total effect.
We show that the set of stack signatures forms an ordered monoid,
\textit{i.e.}, a monoid that is equipped with a partial order compatible with the 
multiplication of the monoid. 
We then formulate pushdown systems as automata over the monoid of stack
signatures.

We extend the structure of weight by introducing semirings indexed
by a monoid element. 
An indexed semiring $\mathcal{S}$ over a monoid $\mathcal{M}$ has domains $D_m$ indexed by
$m \in \mathcal{M}$ and indexed operations $\otimes_{m,m'} : D_{m} \times D_{m'} \to D_{m m'}$ and $\oplus_m: D_m \times D_m \to D_m$ for $m,m' \in\mathcal{M}$.
The operations must satisfy the properties of semirings extended to indexed domains.
Weighted pushdown systems are then generalized to those 
over a semiring indexed by the monoid of stack signatures. 
We show that the reachability analysis of weighted pushdown systems 
by Reps \textit{et al.}~\cite{reps:weighted} can
be refined to those over an indexed semiring and the boundedness can be
replaced with the \emph{local boundedness}.

To prove that a structure forms an indexed semiring, we need to show
many properties on its multiplication and addition.
It is rather cumbersome to prove them from scratch.
We show that an indexed semiring can be constructed from a simplified 
structure, called a \emph{weight structure}.
All the indexed semirings used in our applications of weighted pushdown systems are presented
as  weight structures. It is much easier to show a structure forms
a weight structure.

We present several applications of pushdown systems with indexed weighted domains.
The first application is an encoding of a pushdown system into a weighted 
pushdown system whose stack alphabet is a singleton. 
This is a simplified version of the encoding of local variables into weight
by Suwimonteerabuth~\cite{SuwimonteerabuthPhd}.
The second application is an indexed semiring to encode the reachability analysis of
conditional pushdown systems into that of weighted pushdown systems.
We also consider the coverability in well-structured pushdown systems by 
Cai and Ogawa~\cite{ogawa:well}, 
and the reachability in pushdown systems with stack manipulation
by Uezato and Minamide~\cite{uezato:stack}.
Since the indexed semirings used in these applications are locally bounded, 
our framework guarantees termination of the analyses.

\medskip

This paper is organized as follows.
Section~\ref{sec:wautomata} reviews the definitions of semirings and weighted automata.
In Section~\ref{sec:sig}, we introduce stack signatures that summarize the effect of a computation of a pushdown system and show that they form a semiring.
In Section~\ref{sec:indexed}, we introduce semirings indexed by a monoid and
weighted automata are extended to those over an indexed semiring.
Section~\ref{sec:wpds} introduces weighted pushdown automata over
an indexed semiring and extends the standard saturation procedure to them.
Section~\ref{sec:wstructure} presents a simplified structure to easily 
construct a semiring indexed by a monoid.
Several applications of our framework are presented in Section~\ref{sec:application}.
Finally, we discuss related work and conclude.

\section{Semirings and Weighted Automata} \label{sec:wautomata}
We first review the definitions of semirings and weighted automata.
\begin{defi}
A semiring is a structure $\mathcal{S}=\mytup{D,\oplus,\otimes,0,1}$ where
$D$ is a set, $0$ and $1$ are elements of $D$, $\oplus$ and $\otimes$ are
binary operations on $D$ such that
\begin{enumerate}
\item $\tup{D, \oplus, 0}$ is a commutative monoid.
\item $\tup{D, \otimes, 1}$ is a monoid.
\item $\otimes$ distributes over $\oplus$.
\[ \begin{array}{c@{\quad\quad}c}
(x \oplus y) \otimes z = (x  \otimes z) \oplus (y  \otimes z)  &
 x \otimes (y \oplus z)  = (x  \otimes y) \oplus (x  \otimes z)  
\end{array} \]
\item $0$ is an annihilator with respect to $\otimes$: 
$0 \otimes x = 0 = x \otimes 0$ for all $x \in D$.
\end{enumerate}
\end{defi}
We say that a semiring $\mathcal{S}$ is \emph{idempotent} if its addition $\oplus$ is idempotent (i.e., $a \oplus a = a$). For an idempotent semiring 
$\mytup{D,\oplus,\otimes,0,1}$, $\tup{D, \oplus}$ can be
considered as a join semilattice\footnote{In~\cite{reps:weighted}, it is considered as a meet semilattice.}. 
Then, the partial order $\sqsubseteq$ is defined 
by $a \sqsubseteq b$ iff $a \oplus b = b$ for an idempotent semiring. 
We say that an idempotent semiring is \emph{bounded} if there are no infinite
ascending chains with respect to $\sqsubseteq$.

In this paper, we consider weighted automata without initial and
final states.
\begin{defi}
A weighted automaton $\mathcal{A}$ over an idempotent semiring $\mathcal{S}$ and
an alphabet $\Gamma$ is 
a structure $\tup{\Gamma, Q, E}$ 
where $Q$ is a finite set of states,
 $E : Q \times \Gamma \times Q \to \mathcal{S}$
is a set of transition rules each of which associates an element in $\mathcal{S}$ as weight.
\end{defi}

For weighted automata over an alphabet $\Gamma$ and
a semiring $\mathcal{S}=\mytup{D,\oplus,\otimes,0,1}$,
we introduce the transition relation of the form $\atrans{q}{q'}{w}{a}$ 
where $w \in \Gamma^*$ and $a \in D$. It is inductively defined as follows.
\begin{itemize}
\item $\atrans{q}{q}{\epsilon}{1}$ for any $q \in Q$.
\item $\atrans{q}{q'}{\gamma}{a}$ if $a = E(\tup{q, \gamma, q'})$.
\item $\atrans{q}{q'}{ww'}{a \otimes b}$ if $\atrans{q}{q''}{w}{a}$ and
$\atrans{q''}{q'}{w'}{b}$.
\end{itemize}

Then, for two states $q$ and $q'$ and a word $w$, 
we consider the total weight of the transitions of the form $\atrans{q}{q'}{w}{a}$ defined as follows\footnote{This is basically a formal power series, which is used to define the behaviour of weighted automata~\cite{weightedautomata}.}.
\[ \delta(q, w, q') = \bigoplus \set{a \mid \atrans{q}{q'}{w}{a}} \]
This is well-defined because there are only finitely many transitions of this form and we assume that the semiring is idempotent.
In the general theory of weighted automata, we do not impose that
the semiring is idempotent~\cite{weightedautomata}. However, we impose the condition to adopt the simple 
and intuitive definition above.

\section{Stack Signatures\label{sec:sig}} 
We introduce stack signatures that summarize the effect of a transition
on stack as a pair of words over a stack alphabet.
It is shown that the set of stack signatures forms 
a monoid, and then a semiring by introducing a partial order on them. 
Stack signatures naturally appear in the theory of context-free
grammars and pushdown systems~\cite{SuwimonteerabuthPhd,minamide06,tozawa07}. 
We adopt the term `stack signature' introduced by Suwimonteerabuth~\cite{SuwimonteerabuthPhd}. 

The proofs of most results in this section appear in Appendix~\ref{appendix:assoc}.
They are not fundamentally difficult, but require detailed case-analysis.
Thus, we also formalized stack signatures and proved 
their properties in Isabelle/HOL by extending our previous work on a formalization of decision procedures on context-free grammars~\cite{minamide07}\footnote{The proof script can be found at \url{http://www.is.titech.ac.jp/~minamide/stacksig.tar.gz}.}.

The effect of a transition of a pushdown system can be summarized as
a pair of sequences of stack symbols written $w_1/w_2$ where $w_1$
are the symbols popped by the transition and $w_2$ are those
pushed by the transition. We consider that pushing $\gamma$ and then
popping the same $\gamma$ cancel the effect, but popping $\gamma$ and 
then pushing $\gamma$ have the effect $\gamma/\gamma$. 
\begin{defi} 
We call elements of $\Gamma^* \times \Gamma^*$ \emph{stack signatures}
  and write $w/w'$ for a stack signature $\tup{w, w'}$.
\begin{itemize}
\item We say that $w_1/w_1'$ and $w_2/w_2'$ are \emph{compatible} if either $w_1'$ is
a prefix of $w_2$ or $w_2$ is a prefix of $w_1'$. 
Furthermore, they are called \emph{strictly compatible} if $w_1' = w_2$.
\item For compatible $w_1/w_1'$ and $w_2/w_2'$, we define $w_1/w_1' \cdot w_2/w_2'$ by
\[ w_1/w_1' \cdot w_2/w_2' =
\left\{ \begin{array}{l@{\quad\quad}l}
w_1/w_2'w_1'' & \mbox{if $w_1' = w_2w_1''$} \\
w_1w_2''/w_2' & \mbox{if $w_2 = w_1'w_2''$}
            \end{array}\right.
\]
\end{itemize}
\end{defi}

\noindent For example, we have $\gamma_1/\gamma_2\cdot \gamma_2\gamma_3/\gamma_4 = 
\gamma_1\gamma_3/\gamma_4$.
We write $\sigma_1 \parallel \sigma_2$ if stack signatures $\sigma_1$ and $\sigma_2$ are strictly compatible.

By introducing an element $\top$ and extending the definition $\cdot$ as follows,
$\mytup{(\Gamma^* \times \Gamma^*) \cup \set{\top}, \cdot, \epsilon/\epsilon}$ forms a monoid. 
The proof of the associativity of $\cdot$ appears in Appendix~\ref{appendix:assoc}.
We write $\mathcal{M}_{\Gamma}$ for this monoid. 
\[ 
\begin{array}{ll}
\top \cdot \sigma = \sigma \cdot \top = \top & \mbox{for $\sigma \in \mathcal{M}_\Gamma$} \\
w_1/w_1' \cdot w_2/w_2' = \top & \mbox{if $w_1/w_1'$ and $w_2/w_2'$ are not compatible}
\end{array} \]
By relaxing the use of terminology, we call an element of $\mathcal{M}_\Gamma$
a \emph{stack signature} and an element of the form $w/w'$ a \emph{proper stack signature}.

The following isomorphism is used to relate automata and pushdown systems. 
It is clear from $w_1/\epsilon \cdot w_2/\epsilon = w_1w_2/\epsilon$.
\begin{prop} \label{prop:submonoid}
The set $\set{w/\epsilon \mid w \in \Gamma^*}$ is a submonoid of $\mathcal{M}_\Gamma$. 
Furthermore, it is isomorphic to $\Gamma^*$ by the function projecting $w$ from $w/\epsilon$.
\end{prop}

We also introduce a partial order on stack signatures:
a transition that pops $w_1$ and pushes $w_2$ can be considered
as one that pops $w_1w$ and pushes $w_2w$ for any $w \in \Gamma^*$.
\begin{defi} 
A partial order $\le$ on stack signatures is 
defined by $w_1/w_2 \le w_1w/w_2w$ for $w_1,w_2,w \in \Gamma^*$
and $\sigma \le \top$ for any stack signature $\sigma$. 
\end{defi}
It is clear that $(\Gamma^* \times \Gamma^*) \cup \set{\top}$ is a join-semilattice.  
This partial order is compatible with the binary operation $\cdot$: if $\sigma_1 \le \sigma_1'$ and $\sigma_2 \le \sigma_2'$, then $\sigma_1 \cdot \sigma_2 \le \sigma_1' \cdot \sigma_2'$ (Lemma~\ref{lem:ordered} in the appendix).
Thus, the monoid of stack signatures is an 
\emph{ordered monoid}\footnote{A monoid is ordered when it is equipped with a compatible partial order.}.
With this order, the compatibility of stack signatures can be understood 
by the strict compatibility.
\begin{lem} Two stack signatures $\sigma_1$ and $\sigma_2$ are compatible if and only if
one of the following holds.
\begin{itemize}
\item $\sigma_1 \le \sigma_1'$ and $\sigma_1' \parallel \sigma_2$
for some $\sigma_1'$.
\item $\sigma_2 \le \sigma_2'$ and $\sigma_1 \parallel \sigma_2'$
for some $\sigma_2'$.
\end{itemize}
\end{lem}
For example, $\gamma_1\gamma_2/\gamma_3$ and $\gamma_3\gamma_4/\gamma_5$ are compatible
because $\gamma_1\gamma_2/\gamma_3 \le \gamma_1\gamma_2\gamma_4/\gamma_3\gamma_4$
and $\gamma_1\gamma_2\gamma_4/\gamma_3\gamma_4 \parallel \gamma_3\gamma_4/\gamma_5$. 
Then, $\cdot$ on compatible stack signatures can also be understood by
$\cdot$ on strictly compatible stack signatures.
\begin{lem} \mbox{}

\begin{itemize}
\item If $\sigma_1 \le \sigma_1'$ and $\sigma_1' \parallel \sigma_2$, 
then $\sigma_1 \cdot \sigma_2 = \sigma_1' \cdot \sigma_2$.
\item If $\sigma_2 \le \sigma_2'$ and $\sigma_1 \parallel \sigma_2'$, 
then $\sigma_1 \cdot \sigma_2 = \sigma_1 \cdot \sigma_2'$.
\end{itemize}
\end{lem}

Furthermore, we can construct an idempotent semiring
by introducing the bottom element $\bot$ and
extending $\cdot$ for $\bot$ as follows. 
\[ \bot  \cdot x = x \cdot \bot = \bot \quad\quad
\mbox{for all $x \in (\Gamma^* \times \Gamma^*) \cup \set{\top, \bot}$} \]
\begin{prop}
Let $S=(\Gamma^* \times \Gamma^*) \cup \set{\top, \bot}$.
$\tup{S, \sqcup, \cdot, \bot, \epsilon/\epsilon}$
forms an idempotent semiring.
\end{prop}
The distributivity of $\cdot$ over $\sqcup$ is proved in Lemma~\ref{lem:distrib}.
This semiring is not bounded because $\epsilon/\epsilon \le
\gamma/\gamma \le \gamma\gamma/\gamma\gamma \le \cdots$.

\section{Semirings Indexed by a Monoid}\label{sec:indexed}
We introduce a semiring indexed by a monoid, which is a typed algebraic
structure where a type is an element of a monoid. Weighted pushdown systems
are generalized by taking this structure as the weight domain in the next section.

\begin{defi} 
Let $\mathcal{M} = \tup{M, \cdot, 1_{\mathcal{M}}}$ be a monoid.
An indexed semiring $\mathcal{S}$ over $\mathcal{M}$ is a structure 
$\tup{\set{D_m}, \set{\oplus_m}, \set{\otimes_{m_1,m_2}}, \set{0_m}, 1}$ such that
\begin{itemize}
\item $D_m$ is a set for each $m \in M$.
\item $\tup{D_m, \oplus_m, 0_m}$ is a commutative monoid for $m\in M$.
\item $\otimes_{m_1,m_2}$ is an associative binary operation of type $D_{m_1} \times D_{m_2} \to D_{m_1m_2}$ for $m_1, m_2 \in M$.
\[ (a \otimes_{m_1,m_2} b) \otimes_{m_1m_2, m_3} c = a \otimes_{m_1,m_2m_3} 
(b \otimes_{m_2,m_3} c) \]
\item $1 \in D_{1_\mathcal{M}}$ is a neutral element of $\otimes_{m,m'}$: $a \otimes_{m,1_{\mathcal{M}}} 1 = 1 \otimes_{1_{\mathcal{M}}, m} a = a$.
\item $\otimes_{m_1,m_2}$ distributes over $\oplus_m$.
\[ (a \oplus_{m_1} b) \otimes_{m_1,m_2} c = (a  \otimes_{m_1,m_2} c) \oplus_{m_1m_2} (b  \otimes_{m_1,m_2} c)  \]
\[ a \otimes_{m_1,m_2} (b \oplus_{m_2} c)  = (a  \otimes_{m_1,m_2} b) \oplus_{m_1m_2} (a  \otimes_{m_1,m_2} c)  \]
\item $0_m$ is an annihilator with respect to $\otimes_{m, m'}$.
\[ 0_{m_1} \otimes_{m_1, m_2} a = 0_{m_1m_2} = b \otimes_{m_1,m_2} 0_{m_2} \]
\end{itemize}
\end{defi}\medskip

\noindent We call $\mathcal{S}$ an idempotent indexed semiring if $\mathcal{S}$ is
an indexed semiring where $\oplus_m$ is idempotent for all $m \in M$.
We introduce partial orders $\sqsubseteq_m$ defined by $a \sqsubseteq_m b$ iff $a \oplus_m b = b$. From distributivity of $\otimes$, it is clear that
$\otimes$ is monotonic with respect to $\sqsubseteq_m$.
If we ignore the monoid structure of each $D_m$, this structure corresponds to a lax monoidal functor $F: \mathcal{M} \to (\mathrm{Set}, \times, \set{*})$ in category theory.

\begin{exa}
Matrices over a semiring have a similar structure, but are indexed by a subgroup instead of a monoid.
Let us consider $m \times n$ matrices over an arbitrary semiring. 
We write $\tup{m,n}$ for the dimensions of $m \times n$ matrices.
Then, the set of dimensions forms a subgroup by introducing $\top$ and defining the binary operation $\cdot$ as 
follows.
\[ \tup{m_1, n_1} \cdot \tup{m_2, n_2} = 
   \left\{
   \begin{array}{ll}
   \tup{m_1, n_2} & \mbox{if $n_1 = m_2$}  \\    
   \top           & \mbox{otherwise}      
   \end{array}
   \right.
\]
Let $D_{\tup{m,n}}$ be the set of $m \times n$ matrices. Then, $D_{\tup{m,n}}$ with
matrix addition and multiplication forms a semiring indexed by the subgroup of dimensions
where $D_\top$ is defined as a singleton.
For boolean matrices, the indexed semiring is idempotent since the addition of
boolean matrices is idempotent. \qed
\end{exa}

The following proposition is used later to consider a semiring indexed by
a submonoid of the stack signatures. The conditions of an indexed semiring
carry over to the substructure.
\begin{prop} \label{prop:subsemiring}
Let $\mathcal{M} = \tup{M, \cdot, 1_{\mathcal{M}}}$ be a monoid and
$\mathcal{S}$ a semiring indexed by $\mathcal{M}$.
If $\mathcal{M}'$ is a submonoid of $\mathcal{M}$, then the restriction
of $\mathcal{S}$ on $\mathcal{M}'$ is a semiring indexed by $\mathcal{M}'$.
\end{prop}

The notion of weighted automata can be extended for an indexed semiring
over the monoid $\Gamma^*$ in the straightforward manner.
\begin{defi} \label{def:indexedwautomata}
Let $\mathcal{S}$ be an idempotent semiring $\tup{\set{D_w}, \set{\oplus_w}, \set{\otimes_{w_1,w_2}}, \set{0_w}, 1}$ indexed by $\Gamma^*$.
A weighted automaton $\mathcal{A}$ over $\mathcal{S}$ is 
a structure $\tup{\Gamma, Q, E}$ 
where $Q$ is a finite set of states, and $E : Q \times \Gamma \times Q \to 
\bigcup_{\gamma\in \Gamma} D_\gamma$ is a set of transition rules assigning
a weight
such that $E(\tup{q,\gamma,q'}) \in D_\gamma$.
\end{defi}

The definition of the transition relation is revised as follows. The only revision
is that we apply indexed $\otimes_{w,w'}$ to combine two transitions for $w$ and $w'$.
\begin{itemize}
\item $\atrans{q}{q}{\epsilon}{1}$ for any $q \in Q$.
\item $\atrans{q}{q'}{\gamma}{a}$ if $a = E(\tup{q, \gamma, q'})$.
\item $\atrans{q}{q'}{ww'}{a \otimes_{w,w'} b}$ if $\atrans{q}{q''}{w}{a}$ and
$\atrans{q''}{q'}{w'}{b}$.
\end{itemize}

\section{Weighted Pushdown Systems over an Indexed Semiring and
Their Reachability Analysis} \label{sec:wpds}
We introduce weighted pushdown systems over a semiring indexed by
the monoid of stack signatures. 
The (generalized) reachability analysis of weighted pushdown systems is refined to 
those over an indexed semiring and the boundedness is relaxed to 
the local boundedness.
We also show that
it is possible to construct an ordinary semiring from an indexed semiring,
but the obtained semiring is not bounded.

\subsection{Weighted Pushdown Systems over an Indexed Semiring} 
We basically consider pushdown systems over a stack alphabet $\Gamma$ as 
automata over the monoid of stack signatures $\mathcal{M}_\Gamma$. 
However, to clarify our presentation we introduce the definition
of weighted pushdown systems independently.
Weight domains $D_\sigma$ are indexed by a stack signature $\sigma$
and forms an indexed semiring over $\mathcal{M}_\Gamma$.

\begin{defi}
Let $\mathcal{S}=\tup{\set{D_\sigma}, \set{\oplus_\sigma}, \set{\otimes_{\sigma_1,\sigma_2}}, \set{0_\sigma}, 1}$ be a semiring indexed by $\mathcal{M}_\Gamma$.
A weighted pushdown system $\mathcal{P}$ over $\mathcal{S}$ is a structure
$\tup{P, \Gamma, \Delta}$ 
where $P$ is a finite set of states, 
$\Gamma$ is a stack alphabet, 
and $\Delta \subseteq P \times  \Gamma \times P \times \Gamma^* \times 
\bigcup_{\gamma\in\Gamma, w\in\Gamma^*} D_{\gamma/w}$
is a finite set of transitions such that $a \in D_{\gamma/w}$ for $\tup{p,\gamma,p',w,a} \in \Delta$.
\end{defi}
A configuration of a pushdown system $\mathcal{P}$ is a pair $\tup{p, w}$ for $p \in P$ and $w \in \Gamma^*$. 
We write $\ptrans{p}{\gamma}{p'}{w}{a}$ if $\tup{p, \gamma, p', w, a} \in  \Delta$. 

We consider pushdown systems as automata over stack signatures and 
define the translation relation as follows:
\begin{itemize}
\item $\atransrule{p}{p}{\epsilon/\epsilon}{1}$.
\item $\atransrule{p}{p'}{\gamma/w}{a}$ if $\ptrans{p}{\gamma}{p'}{w}{a}$.
\item $\atransrule{p}{p'}{\sigma_1 \cdot \sigma_2}{a}$
if $\atransrule{p}{p''}{\sigma_1}{a_1}$, 
$\atransrule{p''}{p'}{\sigma_2}{a_2}$, $a = a_1 \otimes_{\sigma_1,\sigma_2} a_2$ and $\sigma_1 \cdot \sigma_2 \neq \top$.
\end{itemize}
Then, it is clear that $a \in D_\sigma$ if $\atransrule{p}{p'}{\sigma}{a}$.

Traditionally, the transition relation on a pushdown system is defined as
a relation between configurations. To introduce such a definition,
we need to extend an indexed semiring with an additional operation.
\begin{defi} \label{def:conv}
Let $\mathcal{M}$ be an \emph{ordered} monoid with partial order $\le$.
By an indexed semiring over $\mathcal{M}$ we shall mean 
an indexed semiring $\mathcal{S}$ over $\mathcal{M}$ on which there is
a family of conversion functions $\extop{m,m'} :
D_m \to D_{m'}$ indexed by pairs of monoid elements $m \le m'$ such that
\begin{enumerate}
\item $\extop{m,m} = \mathrm{id}$.
\item $\extop{m,m''} =  \extop{m',m''} \circ \extop{m,m'}$ for all $m \le m' \le m''$.
\item \label{conv:prop4}
$\ext{m,m'}{0_m} = 0_{m'}$ and
$\ext{m,m'}{a \oplus_{m} b} = \ext{m,m'}{a} \, \oplus_{m'} \ext{m,m'}{b}$.
\item $\ext{m_1 m_2, m_1'm_2'}{a \otimes_{m_1,m_2} b} = \ext{m_1,m_1'}{a} \otimes_{m_1',m_2'} \ext{m_2,m_2'}{b}$ for all $m_1 \le m_1'$ and  $m_2 \le m_2'$.
\end{enumerate}
\end{defi}

\begin{exa}  \label{exa:hweight}
The structure $\mathcal{S}=\tup{\set{D_\sigma}, \set{\oplus_\sigma}, \set{\otimes_{\sigma,\sigma_2}}, \set{0_\sigma}, 0}$ forms a semiring indexed by the ordered monoid of 
stack signatures. 
\begin{itemize}
\item $D_{w/w'} = \mathbb{N}^{{\ge} \mathrm{max}(|w|,|w'|)} \cup \set{\infty}$ and $D_\top = \set{\infty}$ where $\mathbb{N}^{{\ge} i} = \set{ j \in \mathbb{N} \mid j \ge i}$.
\item $a \oplus_\sigma b = \mathrm{min}(a,b)$ and $0_\sigma = \infty$.
\item $\otimes_{\sigma_1,\sigma_2}$ is defined 
  for compatible $\sigma_1$ and $\sigma_2$ as follows.
\[ a \otimes_{w_1/w_1',w_2/w_2'} b = \left\{
\begin{array}{ll}
\mathrm{max}(|w_2|-|w_1'|+a, b)  & \mbox{if $|w_1'| \le |w_2|$} \\
\mathrm{max}(a, |w_1'|-|w_2|+b) & \mbox{if $|w_2| \le |w_1'|$} 
\end{array} \right. \]
\item The conversion functions are defined by $\ext{w_1/w_1',w_1w/w_2'w}{a} = a+|w|$.
\end{itemize}
It is shown in Example~\ref{exa:heightalg}
that the structure $\mathcal{S}$ really satisfies the conditions of
indexed semirings through the construction introduced in Section~\ref{sec:wstructure}.
This indexed semiring is used to compute the minimum height of transitions 
between two configurations of a pushdown system in Example~\ref{exa:height}. \qed
\end{exa}

For an indexed semiring over the ordered monoid $\mathcal{M}_\Gamma$, 
we write $\extop{w}$ for $\extop{w_1/w_2,w_1w/w_2w}$ if $w_1$ and $w_2$ are
clear from the context.
Then, the standard definition of the transition relation of a weighted pushdown system
is given as follows.
\begin{itemize}
\item $\utransrule{\tup{p,w}}{\tup{p,w}}{\ext{w}{1}}{}$.
\item $\utransrule{\tup{p,\gamma w'}}{\tup{p',ww'}}{\ext{w'}{a}}{}$ if $\ptrans{p}{\gamma}{p'}{w}{a}$.
\item $\utransrule{\tup{p,w}}{\tup{p',w'}}{a}{}$
if $\utransrule{\tup{p,w}}{\tup{p'',w''}}{a_1}{}$, 
$\utransrule{\tup{p'',w''}}{\tup{p',w'}}{a_2}{}$, 
and $a = a_1 \otimes_{w/w'',w''/w'} a_2$.
\end{itemize}

Then, these two definitions of transition relations are equivalent in the following sense.
As a special case of this proposition, we have
$\utransrule{\tup{p,w}}{\tup{p',\epsilon}}{a}{}$ iff
$\atransrule{p}{p'}{w/\epsilon}{a}$.
\begin{prop} \label{prop:conv}
If $\utransrule{\tup{p,w}}{\tup{p',w'}}{a}{}$, then there exist $\sigma$
and $a'$
such that $\sigma \le w/w'$, $\atransrule{p}{p'}{\sigma}{a'}$, and
$a = \ext{\sigma,w/w'}{a'}$.
Conversely, if $\atransrule{p}{p'}{\sigma}{a'}$, then 
$\utransrule{\tup{p,w}}{\tup{p',w'}}{\ext{\sigma,w/w'}{a'}}{}$ for all $\sigma \le w/w'$.
\end{prop}
\begin{proof}
We prove the first direction by induction on the derivation of
$\utransrule{\tup{p,w}}{\tup{p',w'}}{a}{}$.
\begin{description}
\item[Case] $\utransrule{\tup{p,w}}{\tup{p,w}}{\ext{w}{1}}{}$. We have $\atransrule{p}{p}{\epsilon/\epsilon}{1}$, $\epsilon/\epsilon \le w/w$, and
$\ext{w}{1} = \ext{\epsilon/\epsilon,w/w}{1}$.

\item[Case] $\utransrule{\tup{p,\gamma w'}}{\tup{p',ww'}}{\ext{w'}{a}}{}$.
We have $\atransrule{p}{p'}{\gamma/w}{a}$ and $\gamma/w \le \gamma w'/w w'$.
\item[Case]
$\utransrule{\tup{p,w}}{\tup{p',w'}}{a}{}$ is obtained from
$\utransrule{\tup{p,w}}{\tup{p'',w''}}{a_1}{}$, 
$\utransrule{\tup{p'',w'}}{\tup{p',w'}}{a_2}{}$, 
and $a = a_1 \otimes_{w/w'',w''/w'} a_2$. 
By the induction hypothesis, we have 
\begin{itemize}
\item $\atransrule{p}{p''}{\sigma_1}{a_1'}$, $\sigma_1 \le w/w''$, and 
$\ext{\sigma_1,w/w''}{a_1'} = a_1$, 
\item $\atransrule{p''}{p'}{\sigma_2}{a_2'}$, $\sigma_2 \le w''/w'$, and $\ext{\sigma_2,w''/w'}{a_2'} = a_2$.
\end{itemize}
By monotonicity of $\cdot$, $\sigma_1 \cdot \sigma_2 \le w/w'$
and then $\atransrule{p}{p'}{\sigma_1\cdot\sigma_2}{a'}$
where $a' = a_1' \otimes_{\sigma_1,\sigma_2} a_2'$.
We also have 
$\ext{\sigma_1\cdot\sigma_2, w/w'}{a_1' \otimes_{\sigma_1,\sigma_2} a_2'} = 
  \ext{\sigma_1, w/w''}{a_1'} \otimes_{w/w'',w''/w'}
  \ext{\sigma_2, w''/w'}{a_2'} = a$.
\end{description}

\noindent The other direction is proved in a similar manner by induction on the derivation of $\atransrule{p}{p'}{\sigma}{a'}$.

\end{proof}

\subsection{Reachability Analysis}
We show that the reachability analysis of weighted pushdown systems 
by Reps~\textit{et al.}~\cite{reps:weighted} can be
generalized for those over an indexed semiring, where we adopt a localized 
version of the boundedness of a semiring. 
\begin{defi}
We say an indexed idempotent semiring over $\mathcal{M}_{\Gamma}$ is \emph{locally bounded} if $D_{\gamma/\epsilon}$ is bounded for all $\gamma \in \Gamma$.
\end{defi}

First, we focus on the (generalized) backward reachability of a configuration with the empty stack
and consider the problem that computes the following function:
\[ \delta(p,w,p') = 
\bigoplus \set{ a \mid \atransrule{p}{p'}{w/\epsilon}{a}} \]
where the above addition is the extension of $\oplus_{w/\epsilon}$ for a set. 
This function is well-defined if the indexed semiring is locally bounded.
It is clear from the following equation:
\[ \delta(p,\gamma w',p') = \bigoplus_{p'' \in P } (\delta(p,\gamma,p'') \otimes_{\gamma/\epsilon,w'/\epsilon} \delta(p'',w',p')) \]
where we have $\delta(p,\gamma,p'') \in D_{\gamma/\epsilon}$ for all $p'' \in P$.
Although there are infinitely many transitions of the form 
$\atransrule{p}{p''}{\gamma/\epsilon}{a}$, $\delta(p,\gamma,p'')$ is well-defined because $D_{\gamma/\epsilon}$ is bounded.

\medskip 

We generalize the reachability analysis of weighted pushdown automata 
for those over an indexed semiring. 
The algorithm is a generalization of the saturation procedure 
on $\mathcal{P}$-automata~\cite{bouajjani:reachability,finkel:direct,reps:weighted}.

Let us consider a weighted pushdown system $\mathcal{P}=\tup{P,\Gamma,\Delta}$
over a semiring $\mathcal{S}$ indexed by $\mathcal{M}_\Gamma$.
We apply the procedure to a weighted automaton over the restriction of $\mathcal{S}$ to $\set{w/\epsilon \mid w \in \Gamma^*}$~\footnote{
The restriction of $\mathcal{S}$ to $\set{w/\epsilon \mid w \in \Gamma^*}$ is
a semiring indexed by $\set{w/\epsilon \mid w \in \Gamma^*}$ by Proposition~\ref{prop:submonoid} and
~\ref{prop:subsemiring}.} and start from $\mathcal{A}_0 = \tup{P, \Gamma, E_0}$, 
which has no transitions, \textit{i.e.}, $E_{0}(\tup{p,\gamma,p'}) = 0_{\gamma/\epsilon}$ for all $p,p'\in P$ and $\gamma \in \Gamma$.
Then, the weighted  automaton $\mathcal{A}_{\mathrm{pre}^*}$ representing 
$\delta_{\mathcal{P}}(p,\gamma,p')$ can be obtained by applying
the \emph{saturation rule} for weighted pushdown systems to $\mathcal{A}_0$
until saturation.
The following is the saturation rule of Reps~\textit{et al.} for
the backward reachability analysis adapted to our framework~\cite{reps:weighted}.
\begin{itemize}
\item If $\ptrans{p}{\gamma}{p'}{w}{a_1}$ and $\atrans{p'}{p''}{w}{a_2}$ in the current automaton,
  add a transition rule $\atrans{p}{p''}{\gamma}{a}$ to the automaton where $a = a_1 \otimes_{\gamma/w,w/\epsilon} a_2$.
\end{itemize}
When we add $\atrans{p}{p''}{\gamma}{a}$, if there already exists
transition $\atrans{p}{p''}{\gamma}{a'}$, then we replace it
with $\atrans{p}{p''}{\gamma}{a \oplus_{\gamma/\epsilon} a'}$.

Since there are only finitely many (one-step) transitions in
$\mathcal{A}_{\mathrm{pre}^*}$, it is clear that 
the application of the rule terminates if the indexed semiring 
is locally bounded.
\begin{thm} Let $\mathcal{P}$ be a weighted pushdown system over a locally bounded 
idempotent semiring indexed by $\mathcal{M}_\Gamma$.
\begin{itemize}
\item The saturation procedure above terminates.
\item Let $\mathcal{A}_{\mathrm{pre}^*}$ be a weighted
 automaton obtained by the saturation procedure.
Then, we have 
$\atransx{p}{p'}{\gamma}{a}{\mathcal{A}_{\mathrm{pre}^*}}$ for $a = \delta_\mathcal{P}(p,\gamma,p')$.
\end{itemize}
\end{thm}\smallskip

\noindent As a corollary, we have $\atransx{p}{p'}{w}{a}{\mathcal{A}_{\mathrm{pre}^*}}$ for $a = \delta_{\mathcal{P}}(p,w,p')$. 
Before the proof of the theorem, we illustrate the saturation procedure by an example.

\begin{exa} \label{exa:height}
The minimum height of transitions between two configurations can be
computed by the indexed semiring of Example~\ref{exa:hweight}.
Let $\mathcal{P} = \tup{P, \Gamma, \Delta}$ be an ordinary pushdown system.
For a computation $\mathcal{C}: \tup{p_1,w_1} \Longrightarrow \tup{p_2,w_2} \Longrightarrow \cdots \Longrightarrow \tup{p_n,w_n}$ of $\mathcal{P}$, the height of $\mathcal{C}$ is defined by $\sheight{\mathcal{C}} = \mathrm{max}_{1\le i \le n}|w_i|$.
We then consider the minimum height of computations between two 
configurations.

The minimum height can be determined by the reachability analysis
of the weighted pushdown system $\mathcal{P}' = \tup{P, \Gamma, \Delta'}$
where $\Delta'$ is given by:
$\tup{p,\gamma,p',w, \mathrm{max}(1,|w|)} \in \Delta'$ if
$\tup{p,\gamma,p',w} \in \Delta$. Then, we have the following transitions in $\mathcal{P}'$.

\begin{itemize}
\item For a transition with no real moves, $\utransrule{\tup{p,w}}{\tup{p,w}}{\extop{\epsilon/\epsilon,w/w}(0)}{\mathcal{P}'}$ where $\extop{\epsilon/\epsilon,w/w}(0) = |w|$. 
\item For a one-step transition for $\tup{p_1,\gamma,p_2,w} \in \Delta$, we have
\[ \utransrule{\tup{p_1,\gamma w'}}{\tup{p_2,w w'}}{\extop{\gamma/w,\gamma w'/ww'}(\mathrm{max}(1,|w|))}{\mathcal{P}'} \]
where
$\extop{\gamma/w,\gamma w'/ww'}(\mathrm{max}(1,|w|) = \mathrm{max}(1,|w|) + |w'| =
\mathrm{max}(|\gamma w'|,|ww'|)$.
\item 
For $\utransrule{\tup{p_1,w_1}}{\tup{p_2,w_2}}{n_1}{\mathcal{P}'}$
and $\utransrule{\tup{p_2,w_2}}{\tup{p_3,w_3}}{n_2}{\mathcal{P}'}$, 
we have $\utransrule{\tup{p_1,w_1}}{\tup{p_3,w_3}}{\mathrm{max}(n_1,n_2)}{\mathcal{P}'}$.
\end{itemize}
Thus, we can compute the minimum height of computations 
by the reachability analysis of $\mathcal{P}'$.

Let us consider the pushdown system $\mathcal{P_{\mathrm{ex}}}$ in Figure~\ref{fig:exa}.
$\mathcal{P_{\mathrm{ex}}}$ is designed so that the following holds.
\[ \begin{array}{rcl}
\tup{p_0,\gamma\gamma^m} \Longrightarrow \tup{p_1,w} & \mbox{iff} &
\mbox{$w = \gamma^{3n+m}$ for some $n > 0$} \\ 
\tup{p_1,w} \Longrightarrow \tup{p_3,\epsilon} & \mbox{iff} &
\mbox{$w = \gamma^{2n}$ for some $n > 0$} \\ 
\end{array} \]
Thus, the minimum height of computations between $\tup{p_0,\gamma}$ and
$\tup{p_3,\epsilon}$ is $6$.

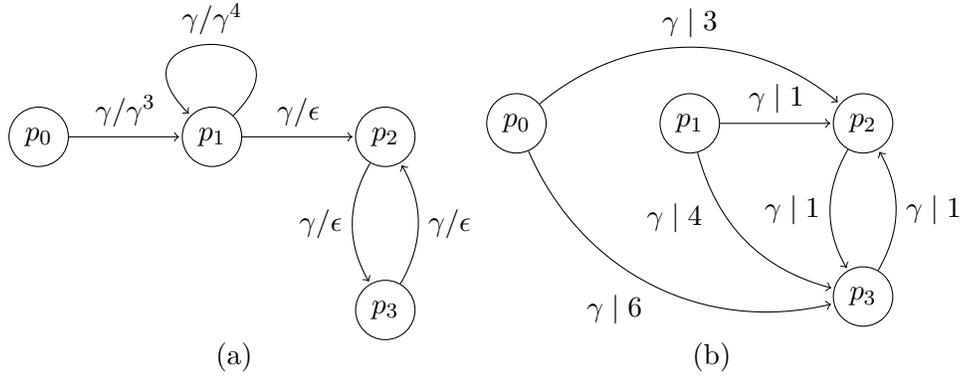
\begin{figure}
\tikzstyle{state without output}=[circle,draw,minimum size=1.5em,every state]

\noindent
\begin{tabular}{cc}
\begin{minipage}{6cm}
\begin{tikzpicture}[shorten >=1pt,node distance=1.5cm,auto]
\node[state]         (p0)               {$p_0$};
\node[state]         (p1) [right=of p0] {$p_1$};
\node[state]         (p2) [right=of p1] {$p_2$};
\node[state]         (p3) [below=of p2] {$p_3$};
\path[->] (p0) edge node {$\gamma/\gamma^3$} (p1)
          (p1) edge [loop] node [swap] {$\gamma/\gamma^4$} (p1)
          (p1) edge node {$\gamma/\epsilon$} (p2)
          (p2) edge [bend right] node [swap] {$\gamma/\epsilon$} (p3)
          (p3) edge [bend right] node [swap] {$\gamma/\epsilon$} (p2);
\end{tikzpicture}
\end{minipage}
&
\begin{minipage}{6cm}
\begin{tikzpicture}[shorten >=1pt,node distance=1.5cm,auto]
\node[state]         (p0)               {$p_0$};
\node[state]         (p1) [right=of p0] {$p_1$};
\node[state]         (p2) [right=of p1] {$p_2$};
\node[state]         (p3) [below=of p2] {$p_3$};
\path[->] 
          
          (p1) edge node {$\gamma \mid 1$} (p2)
          (p2) edge [bend right] node [swap] {$\gamma \mid 1$} (p3)
          (p3) edge [bend right] node [swap] {$\gamma \mid 1$} (p2)

          (p1) edge [bend right] node [pos=0.2,swap] {$\gamma \mid 4$} (p3)
          (p0) edge [bend left=40] node {$\gamma \mid 3$} (p2)
          (p0) edge [bend right=40] node [swap] {$\gamma \mid 6$} (p3)

;
\end{tikzpicture}
\end{minipage} \\
(a) & (b)
\end{tabular}
\caption{(a) pushdown system $\mathcal{P_{\mathrm{ex}}}$. \hspace{.5cm} (b) weighted  automaton $\mathcal{A}_{\mathrm{pre}^*}$ of $\mathcal{P_{\mathrm{ex}}}$.} 
\label{fig:exa}
\end{figure}

Let us determine this by the reachability analysis of $\mathcal{P}_{\mathrm{ex}}'$.
We apply the saturation procedure to $\mathcal{P}_{\mathrm{ex}}'$.  
\begin{enumerate}
\item From $\ptrans{p_1}{\gamma}{p_2}{\epsilon}{1}$ and $\atrans{p_2}{p_2}{\epsilon}{0}$, we add $\atrans{p_1}{p_2}{\gamma}{a_1}$ where 
$a_1 = 1 \otimes_{\gamma/\epsilon,\epsilon/\epsilon} 0 = \mathrm{max}(1,0) = 1$.
Similarly, we add $\atrans{p_2}{p_3}{\gamma}{1}$ and $\atrans{p_3}{p_2}{\gamma}{1}$.
\item From $\atrans{p_1}{p_2}{\gamma}{1}$ and $\atrans{p_2}{p_3}{\gamma}{1}$,
we have $\atrans{p_1}{p_3}{\gamma^2}{a_2}$ where 
$a_2 = 1 \otimes_{\gamma/\epsilon,\gamma/\epsilon} 1 = \mathrm{max}(1+1,1) = 2$.
Similarly, we have $\atrans{p_1}{p_2}{\gamma^3}{3}$.

Then, from $\ptrans{p_0}{\gamma}{p_1}{\gamma/\gamma^3}{3}$ and $\atrans{p_1}{p_2}{\gamma^3}{3}$, we add $\atrans{p_0}{p_2}{\gamma}{3}$. 

\item The other two transitions are added in the same manner.
\end{enumerate}

The transition $\atrans{p_0}{p_3}{\gamma}{6}$ in $\mathcal{A}_{\mathrm{pre}^*}$ 
corresponds to the following computation of $\mathcal{P_{\mathrm{ex}}}$.
\[ \tup{p_0,\gamma} \Longrightarrow \tup{p_1,\gamma^3} 
\Longrightarrow \tup{p_1,\gamma^6} 
\Longrightarrow \cdots \Longrightarrow \tup{p_3,\epsilon}\eqno{\qEd}
\]
\end{exa}\medskip

\noindent The theorem is proved from the following two lemmas.

\begin{lem}
If $\atransrulex{p}{p'}{w/\epsilon}{a}{\mathcal{P}}$,
then 
$\atransx{p}{p'}{w}{a'}{\mathcal{A}_{\mathrm{pre}^*}}$ and $a \sqsubseteq_{w/\epsilon} a'$ for some $a'$.
\end{lem}
\proof
If we only consider the transition relation of the form $\atransrulex{p}{p'}{w/\epsilon}{a}{\mathcal{P}}$, it has the following equivalent inductive definition.
\begin{itemize}
\item $\atransrule{p}{p}{\epsilon/\epsilon}{1}$.
\item $\atransrule{p}{p'}{\gamma w/\epsilon}{a}$
if $\ptrans{p}{\gamma}{p''}{w'}{a_1}$, 
$\atransrule{p''}{p'}{w'w/\epsilon}{a_2}$, and $a = a_1 \otimes_{\gamma/w',w'w/\epsilon} a_2$.
\end{itemize}
By induction on the derivation of $\atransrulex{p}{p'}{w/\epsilon}{a}{\mathcal{P}}$ in the above form.
\begin{description}
\item[Case] $\atransrulex{p}{p}{\epsilon/\epsilon}{1}{}$. The claim holds because $\atrans{p}{p}{\epsilon}{1}$.


\item[Case] 
$\atransrule{p}{p'}{\gamma w_2/\epsilon}{a}$ is obtained from
$\ptrans{p}{\gamma}{p''}{w_1}{a_0}$,
$\atransrulex{p''}{p'}{w_1w_2/\epsilon}{a_3}{}$, and 
$a = a_0 \otimes_{\gamma/w_1, w_1w_2/\epsilon} a_3$.
By induction hypothesis,  $\atrans{p''}{p'}{w_1w_2}{a_3'}$ and 
$a_3 \sqsubseteq_{w_1w_2/\epsilon} a_3'$. Then, we have 
\[ \begin{array}{c@{\quad\quad\quad}c}
\atrans{p''}{p'''}{w_1}{a_1'} &
\atrans{p'''}{p'}{w_2}{a_2'} 
\end{array} \]
and $a_3' = a_1' \otimes_{w_1/\epsilon,w_2/\epsilon} a_2'$ for some $p'''$, $a_1'$, and $a_2'$.

Let $\mathcal{A}_{\mathrm{pre}^*} = \tup{P, \Gamma, E_{\mathrm{pre}^*}}$.
By construction of $\mathcal{A}_{\mathrm{pre}^*}$, 
\[ 
a_0 \otimes_{\gamma/w_1, w_1/\epsilon} a_1' \sqsubseteq_{\gamma/\epsilon} E_{\mathrm{pre}^*}(\tup{p,\gamma,p'''})
\]
Hence
\begin{eqnarray*}
a = a_0 \otimes_{\gamma/w_1, w_1w_2/\epsilon} a_3 & \sqsubseteq_{\gamma w_2/\epsilon} &
a_0 \otimes_{\gamma/w_1, w_1w_2/\epsilon} (a_1' \otimes_{w_1/\epsilon,w_2/\epsilon} a_2')  \\
& \sqsubseteq_{\gamma w_2/\epsilon} &
E_{\mathrm{pre}^*}(\tup{p,\gamma,p'''}) \otimes_{\gamma,w_2} a_2' 
\end{eqnarray*}
and 
\[ \atrans{p}{p'}{\gamma w_2}{E_{\mathrm{pre}^*}(\tup{p,\gamma,p'''}) \otimes_{\gamma,w_2} a_2'} \eqno{\qEd}\]
\end{description}\medskip

\noindent Let $\mathcal{A}_{i+1}$ be a weighted automaton obtained by applying the saturation rule once to
$\mathcal{A}_{i}$.
\begin{lem}
If $\atransx{p}{p'}{\gamma}{a}{\mathcal{A}_i}$, then $a \sqsubseteq_{\gamma/\epsilon} \delta_{\mathcal{P}}(p, \gamma, p')$.
\end{lem}
\begin{proof}
By induction on $i$. 
For $i=0$, the statement trivially holds because
$a = 0_{\gamma/\epsilon}$ for $\atransx{p}{p'}{\gamma}{a}{\mathcal{A}_0}$. 
By assuming the case for $i$, we show the case for $i+1$.
We only consider the case where
$\atransx{p}{p'}{\gamma}{a}{\mathcal{A}_{i+1}}$ is added by the last application
of the saturation rule. Let us assume that $\atransx{p}{p'}{\gamma}{a}{\mathcal{A}_{i+1}}$ is added because of
$\ptrans{p}{\gamma}{p''}{w}{a_1}$,
$\atransx{p''}{p'}{w}{a_2}{\mathcal{A}_i}$, 
$\atransx{p}{p'}{\gamma}{a_0}{\mathcal{A}_i}$, and
$a = a_1 \otimes_{\gamma/w,w/\epsilon} a_2 \oplus_{\gamma/\epsilon} a_0$.

By induction hypothesis,
$a_2 \sqsubseteq_{w/\epsilon} \delta_{\mathcal{P}}(p'', w, p')$ and
$a_0 \sqsubseteq_{\gamma/\epsilon} \delta_{\mathcal{P}}(p, \gamma, p')$.
We also have $a_1 \otimes_{\gamma/w,w/\epsilon} \delta_{\mathcal{P}}(p'', w, p') \sqsubseteq_{\gamma/\epsilon} \delta_{\mathcal{P}}(p,\gamma, p')$
from $\ptrans{p}{\gamma}{p''}{w}{a_1}$. Hence,
$a_1 \otimes_{\gamma/w,w/\epsilon} a_2 \sqsubseteq_{\gamma/\epsilon} a_1 \otimes \delta_{\mathcal{P}}(p'', w, p') \sqsubseteq_{\gamma/\epsilon} \delta_{\mathcal{P}}(p,\gamma, p')$.
Thus, $a \sqsubseteq_{\gamma/\epsilon} \delta_{\mathcal{P}}(p, \gamma, p')$. 
\end{proof}

\subsection{Reachability to a Regular Set of Configurations}
In previous works of the reachability analysis of pushdown systems, it is 
common to consider the reachability problem to a regular set of 
configurations. For a weighted pushdown automaton over an indexed semiring, this problem must be generalized
for a regular set with weight represented by a weighted automaton.

Let us consider an indexed semiring $\mathcal{S}$ over $\mathcal{M}_\Gamma$
and a weighted pushdown system $\mathcal{P}$ over $\mathcal{S}$.
We also consider a weighted automaton $\mathcal{A}$ over the restriction of $\mathcal{S}$
to $\set{w/\epsilon \mid w \in \Gamma^*}$ with the initial states $q_0$ and
the set of final states $F$.
Without loss of generality, we assume that there are no incoming transitions to
$q_0$.  
For a given state $p'$, 
$\mathcal{A}$ represents the set of configurations $\set{\tup{p',w'} \mid \mbox{$w'$ is accepted by $\mathcal{A}$}}$.
Then, the generalized reachability problem to the regular set of configurations
is to compute the following function\footnote{For simplicity, we consider 
the set of configurations whose state is a fixed $p'$. It is easy to extend 
the discussion for the general case.}.
\[ \delta_{\mathcal{P},\mathcal{A}}(p,w,p') = 
\bigoplus_{q \in F} \set{ a \otimes_{\sigma,w'/\epsilon} a' \mid \mbox{$\atransrulex{p}{p'}{\sigma}{a}{\mathcal{P}}$, 
$\atransx{q_0}{q}{w'}{a'}{\mathcal{A}}$, and $\sigma \cdot w'/\epsilon = w/\epsilon$}} \]
This function can be computed by applying the saturation procedure to
the pushdown system $\mathcal{P}'$ obtained by combining $\mathcal{P}$ and
$\mathcal{A}$ with the identification of $p'$ and $q_0$. 
This corresponds to the saturation procedure using $\mathcal{P}$-automata.

The condition $\sigma \cdot w'/\epsilon = w/\epsilon$ above is equivalent to $\sigma \le w/w'$. Furthermore, if the indexed semiring is equipped with the conversion functions $\extop{\sigma_1, \sigma_2}$, we have the following.
\begin{eqnarray*}
\delta_{\mathcal{P},\mathcal{A}}(p,w,p') & =  &
\bigoplus_{q \in F} \set{ a \otimes_{\sigma,w'/\epsilon} a' \mid \mbox{$\atransrulex{p}{p'}{\sigma}{a}{\mathcal{P}}$, 
$\atransx{q_0}{q}{w'}{a'}{\mathcal{A}}$, and $\sigma \cdot w'/\epsilon = w/\epsilon$}} \\
& = &
\bigoplus_{q \in F} \set{ \ext{\sigma,w/w'}{a} \otimes_{w/w',w'/\epsilon} a' \mid \mbox{$\atransrulex{p}{p'}{\sigma}{a}{\mathcal{P}}$, 
$\atransx{q_0}{q}{w'}{a'}{\mathcal{A}}$, and $\sigma \le w/w'$}} \\
& & \mbox{(by Definition~\ref{def:conv} (\ref{conv:prop4}))}   \\
& = & 
\bigoplus_{q \in F} \set{ a \otimes_{w/w',w'/\epsilon} a' \mid \mbox{
$\utransrule{\tup{p,w}}{\tup{p',w'}}{a}{\mathcal{P}}$ 
and $\atransx{q_0}{q}{w'}{a'}{\mathcal{A}}$}} \\
& & \mbox{(by Proposition~\ref{prop:conv})}   \\
\end{eqnarray*}

\noindent The reason why we need to consider a weighted automaton $\mathcal{A}$ instead of just an
automaton is that $D_{w/\epsilon}$ does not have a neutral element on $\otimes$
in general. Thus, we need to consider $a'$ above.

\subsection{Constructing a Semiring from an indexed Semiring over Stack Signatures}
We show that an ordinary semiring can be constructed from a semiring
indexed by stack signatures.
However, the semiring obtained by the construction is
not bounded even for a locally bounded indexed semiring.
Thus, the standard framework of the reachability analysis of weighted
pushdown systems cannot guarantee termination of the saturation procedure.
Although a similar construction appears in~\cite{SuwimonteerabuthPhd},
the definition of $\oplus$ differs from ours and his construction fails to satisfy
the distributivity of $\otimes$ over $\oplus$.

 Let $\mathcal{S}=\tup{\set{D_\sigma}, \set{\oplus_\sigma}, \set{\otimes_{\sigma_1,\sigma_2}}, \set{0_\sigma}, 1_{\mathcal{S}}, \extop{\sigma,\sigma'}}$ be a semiring indexed by the ordered monoid $\mathcal{M}_\Gamma$.
Then, we define a structure $\tup{D, \oplus, \otimes, \bot, 1}$ as follows.
\begin{itemize}
\item  $D = \bigcup_{\sigma \in \mathcal{M}_\Gamma} \set{\tup{\sigma, a} \mid a \in D_\sigma} \cup \set{\bot}$.
\item $1$ is $\tup{\epsilon/\epsilon, 1_{\mathcal{S}}}$.
\item $\oplus$ is defined by $\bot \oplus x = x  = x \oplus \bot$ for all $x \in D$ and
\[ 
\tup{\sigma_1, a} \oplus \tup{\sigma_2, b} = 
\tup{\sigma_1 \sqcup \sigma_2, \ext{\sigma_1,\sigma_1 \sqcup \sigma_2}{a} \oplus_{\sigma_1\sqcup\sigma_2} \ext{\sigma_2,\sigma_1 \sqcup \sigma_2}{b}}.
\]
\item $\otimes$ is defined by $\tup{\sigma_1, a} \otimes \tup{\sigma_2, b} = \tup{\sigma_1 \cdot \sigma_2, a \otimes_{\sigma_1,\sigma_2} b}$ and $x \otimes \bot = \bot = \bot \otimes x$ for all $x \in D$.
\end{itemize}
\begin{thm} \label{thm:semiring}
$\tup{D, \oplus, \otimes, \bot, 1}$ forms a semiring.
\end{thm}
\proof
We show the associativity of $\oplus$ 
and the distributivity of $\otimes$ over $\oplus$.
\begin{itemize}
\item Associativity of $\oplus$. Let $\sigma = \sigma_1 \sqcup \sigma_2\sqcup\sigma_3$.
\begin{eqnarray*}
(\tup{\sigma_1, a} \oplus \tup{\sigma_2, b}) \oplus \tup{\sigma_3, c} & = &
\tup{\sigma_1 \sqcup \sigma_2, \ext{\sigma_1,\sigma_1 \sqcup \sigma_2}{a} \,\oplus_{\sigma_1\sqcup\sigma_2} \ext{\sigma_2,\sigma_1 \sqcup \sigma_2}{b}} \oplus \tup{\sigma_3, c} \\
& = & \tup{\sigma, \ext{\sigma_1,\sigma}{a} \, \oplus_{\sigma} 
\ext{\sigma_2,\sigma}{b} \, \oplus_{\sigma} \ext{\sigma_3,\sigma}{b}} \\
& = &  \tup{\sigma_1, a} \oplus (\tup{\sigma_2, b} \oplus \tup{\sigma_3, c}) 
\end{eqnarray*}

\item $\otimes$ distributes over $\oplus$. 
Let $\sigma = \sigma_1\cdot \sigma_3 \sqcup \sigma_2\cdot\sigma_3$.
\begin{eqnarray*}
(\tup{\sigma_1, a} \oplus \tup{\sigma_2, b}) \otimes \tup{\sigma_3, c} & = &
\tup{\sigma_1 \sqcup \sigma_2, \ext{\sigma_1,\sigma_1 \sqcup \sigma_2}{a} \,\oplus_{\sigma_1\sqcup\sigma_2} \ext{\sigma_2,\sigma_1 \sqcup \sigma_2}{b}} \otimes \tup{\sigma_3, c} \\
& = & \tup{\sigma, \ext{\sigma_1,\sigma_1\sqcup\sigma_2}{a} \otimes_{\sigma_1\sqcup\sigma_2,\sigma_3} c \, 
    \oplus_{\sigma}
 \ext{\sigma_2,\sigma_1\sqcup\sigma_2}{b} \otimes_{\sigma_1\sqcup\sigma_2,\sigma_3} c}\\
& = &   \tup{\sigma, \ext{\sigma_1\sigma_3,\sigma}{a\otimes_{\sigma_1,\sigma_3}c} 
\, \oplus_{\sigma} 
\ext{\sigma_2\sigma_3,\sigma}{b\otimes_{\sigma_2,\sigma_3}c}} \\
& = &   \tup{\sigma_1\cdot \sigma_3, a\otimes_{\sigma_1,\sigma_3}c} \oplus \tup{\sigma_2\cdot\sigma_3, b \otimes_{\sigma_2,\sigma_3} c}\\
 & = &   (\tup{\sigma_1, a} \otimes \tup{\sigma_3, c}) \oplus
       (\tup{\sigma_2, b} \otimes \tup{\sigma_3, c})\rlap{\hbox to 95 pt{\hfill\qEd}}
\end{eqnarray*}
\end{itemize}\medskip

\noindent The construction also works for any semiring indexed by an 
ordered monoid $\mathcal{M}$ if $\mathcal{M}$ has the join operation $\sqcup$.

Suwimonteerabuth did not consider the partial order on stack signatures and
defined the addition of the semiring $\oplus'$ in the following
manner~\cite{SuwimonteerabuthPhd}:
\[ \begin{array}{l}
\tup{\sigma_1, a} \oplus' \tup{\sigma_2, b} = 
\left\{ \begin{array}{l@{\quad\quad}l}
\tup{\sigma_1, a \oplus_{\sigma_1} b} &
\mbox{if $\sigma_1 =  \sigma_2$} \\
(\top, \bullet) & \mbox{otherwise}
\end{array} \right. \\
\end{array} \]
where we assume $D_{\top} = \set{\bullet}$.
However, $\otimes$ does not distribute over $\oplus'$, and thus his construction fails to
form a semiring.
\[ 
(\tup{\epsilon/\epsilon, a} \oplus' \tup{\gamma/\gamma, b}) \otimes \tup{\gamma/\gamma, c}
=  \tup{\top,\bullet} \otimes \tup{\gamma/\gamma,c} = \tup{\top,\bullet} \]
\begin{eqnarray*} 
(\tup{\epsilon/\epsilon, a} \otimes \tup{\gamma/\gamma, c}) \oplus' (\tup{\gamma/\gamma, b} \otimes \tup{\gamma/\gamma, c}) \hspace*{-1cm} \\
& =  &\tup{\gamma/\gamma, a \otimes_{\epsilon/\epsilon,\gamma/\gamma} c} \oplus'
   \tup{\gamma/\gamma, b \otimes_{\gamma/\gamma,\gamma/\gamma} c} \\
& =  & \tup{\gamma/\gamma, a \otimes_{\epsilon/\epsilon,\gamma/\gamma} c \oplus_{\gamma/\gamma}
b \otimes_{\gamma/\gamma,\gamma/\gamma} c}
\end{eqnarray*}\smallskip\vspace{-4 pt}

\noindent It should be noted that the semiring constructed in Theorem~\ref{thm:semiring}
is not bounded as the following sequence shows.
\[ \tup{\epsilon/\epsilon, a} \sqsubset
\tup{\gamma/\gamma, \ext{\gamma}{a}} \sqsubset
\tup{\gamma\gamma/\gamma\gamma, \ext{\gamma\gamma}{a}} \sqsubset \cdots \]
This is one of the reasons why we refine the formulation of
the reachability analysis of weighted pushdown systems in this paper.

The semiring constructed in Theorem~\ref{thm:semiring} actually has the structure of
 a graded semiring. 
Although a graded structure is usually defined for rings~\cite{algebra}, we apply it to
semirings.
A graded semiring $\tup{D,\oplus,\times,1,0}$ over 
$\mathcal{M}$ is a semiring
where $D = \biguplus_{m \in \mathcal{M}} D_{m}$, $D_m$ is a commutative monoid, and
$D_mD_{m'} \subseteq D_{mm'}$ for all $m, m' \in \mathcal{M}$.
It is clear that the semiring in Theorem~\ref{thm:semiring} is a graded semiring over $\mathcal{M}_\Gamma \cup \set{\bot}$ where 
$D = \biguplus_{\sigma \in \mathcal{M}_\Gamma} D_\sigma' \uplus D_\bot'$, $D_\sigma' = \set{\tup{\sigma, a} \mid a \in D_\sigma}$, and $D_{\bot}' =\set{\bot}$.
Furthermore, $D_\sigma'$ has no infinite ascending chains on
$\sqsubset$ if the indexed semiring is locally bounded.
Thus, it is also possible to present our framework based on graded semirings.

\section{Simplified Structure: Multiplication on Strictly Compatible Signatures} \label{sec:wstructure}
An indexed semiring has a multiplication indexed by two stack signatures.
However, it is often simpler to consider and implement a restricted
multiplication defined only for strictly compatible signatures. 
We show that an indexed semiring over the ordered monoid of stack signatures can be constructed from such a structure.

We introduce \emph{weight structures} that have a restricted multiplication $\odot_{\sigma_1,\sigma_2}$ for strictly compatible $\sigma_1$ and $\sigma_2$.

\begin{defi}
A weight structure $\mathcal{W}$ over a stack alphabet $\Gamma$ is 
$\mytup{\set{D_\sigma}, \set{\oplus_{\sigma}}, \set{\odot_{\sigma_1,\sigma_2}}, \set{0_\sigma}, \set{1_\sigma},
\set{\extop{\sigma,\sigma'}}}$ such that
\begin{itemize}
\item $D_\sigma$ is a set for each proper stack signature $\sigma$.
\item $\tup{D_\sigma, \oplus_{\sigma}, 0_\sigma}$ is a commutative monoid
  for each proper stack signature $\sigma$.
\item $\odot_{\sigma_1,\sigma_2}$ is an associative binary operation of 
$D_{\sigma_1} \times D_{\sigma_2} \to D_{\sigma_1\sigma_2}$ for strictly compatible
signatures $\sigma_1$ and $\sigma_2$.

\item $1_\sigma \in D_\sigma$ is an indexed neutral element for $\epsilon/\epsilon \le \sigma$: $a \, \odot_{\sigma',\sigma} 1_\sigma = a$ and 
$1_\sigma \, \odot_{\sigma,\sigma''} b = b$.

\item $0_\sigma$ is an annihilator 
with respect to $\odot_{\sigma, \sigma'}$:
$0_{\sigma_1} \odot_{\sigma_1, \sigma_2} a = 0_{\sigma_1\sigma_2} = b \, \odot_{\sigma_1,\sigma_2} 0_{\sigma_2}$.

\item $\odot$ distributes over $\oplus$.
\[ \begin{array}{rcl}
(a \oplus_{\sigma_1} b) \odot_{\sigma_1,\sigma_2} c & = & (a  \odot_{\sigma_1,\sigma_2} c) \oplus_{\sigma_1\sigma_2} (b  \odot_{\sigma_1,\sigma_2} c)  \\
a \odot_{\sigma_1,\sigma_2} (b \oplus_{\sigma_2} c)  & = & (a  \odot_{\sigma_1,\sigma_2} b) \oplus_{\sigma_1\sigma_2} (a  \odot_{\sigma_1,\sigma_2} c)  
\end{array} \]

\item $\extop{\sigma,\sigma'}$ is a conversion function of $D_\sigma \to D_{\sigma'}$ for $\sigma \le \sigma'$ such that
\begin{itemize}
\item $\extop{\sigma,\sigma} = \mathrm{id}$ and $\extop{\sigma,\sigma''} =  \extop{\sigma',\sigma''} \circ \extop{\sigma,\sigma'}$ for all $\sigma \le \sigma' \le \sigma''$.
\item $\ext{\sigma,\sigma'}{0_\sigma} = 0_{\sigma'}$ and
$\ext{\sigma,\sigma'}{a \oplus b} = \ext{\sigma,\sigma'}{a} \oplus \ext{\sigma,\sigma'}{b}$
\item $\ext{\sigma_1 \cdot \sigma_2,\sigma_1' \cdot \sigma_2'}{a \odot b} = \ext{\sigma_1,\sigma_1'}{a} \odot \ext{\sigma_2,\sigma_2'}{b}$ for 
$\sigma_1 \le \sigma_1'$, $\sigma_2 \le \sigma_2'$, $\sigma_1$ and $\sigma_2$ are strictly compatible, and $\sigma_1'$ and $\sigma_2'$ are strictly compatible.
\item $\ext{\sigma,\sigma'}{1_\sigma} = 1_{\sigma'}$ for $\epsilon/\epsilon \le \sigma \le \sigma'$. 
\end{itemize}

\end{itemize}

\end{defi}

\noindent We show that the multiplication of an indexed semiring over $\mathcal{M}_\Gamma$
can be obtained from that of a weight structure.
Let $\set{D'_\sigma}$ be a family of $\set{D_\sigma} \cup \set{D_\top}$ where
$D_\top = \set{\bullet}$.
Then, the multiplication on $D'_\sigma$ is defined as follows.
\[ x \, \otimes_{\sigma_1,  \sigma_2} y \! = \!
 \left\{ \! \begin{array}{l@{\quad\quad}l}
\ext{\sigma_1,\sigma_1'}{x}  \odot_{\sigma_1',\sigma_2} y & \mbox{if $\sigma_1 \le \sigma_1'$ and $\scomp{\sigma_1'}{\sigma_2}$} \\
x \, \odot_{\sigma_1,\sigma_2'} \ext{\sigma_2,\sigma_2'}{y}  & \mbox{if $\sigma_2 \le \sigma_2'$ and
$\scomp{\sigma_1}{\sigma_2'}$} \\
\bullet & \mbox{otherwise}
             \end{array}\right.
\]
The other operations are extended for $D_\top$ in a straightforward manner.
Then, we obtain a semiring indexed by the ordered monoid $\mathcal{M}_\Gamma$.
\begin{thm}
Let $\tup{\set{D_\sigma}, \set{\oplus_{\sigma}}, \set{\odot_{\sigma_1,\sigma_2}}, \set{0_\sigma}, \set{1_\sigma}, \set{\extop{\sigma,\sigma'}}}$ be a weight structure.
Then, $\tup{\set{D_\sigma'}, \set{\oplus_{\sigma}}, \set{\otimes_{\sigma_1,\sigma_2}}, \set{0_\sigma}, 1_{\epsilon/\epsilon},\set{\extop{\sigma,\sigma'}}}$ is an indexed semiring over an ordered
monoid $\mathcal{M}_\Gamma$.
\end{thm}

Two key properties of the indexed semiring are proved by the following lemmas.
The other properties are easily proved from the corresponding properties of
a weight structure.

\begin{lem} \label{lem:stackstucture1}
$(a \otimes_{\sigma_1,\sigma_2} b) \otimes_{\sigma_1\sigma_2,\sigma_3} c = 
a \otimes_{\sigma_1,\sigma_2\sigma_3} (b \otimes_{\sigma_2,\sigma_3} c)$.
\end{lem}
\proof
We prove the claim by analyzing the cases where $\sigma_1\sigma_2\sigma_3 \ne \top$
by Lemma~\ref{lem:cases}. The proofs of
two cases are omitted because they are symmetric to other cases.
\begin{description}
\item[Case] $\sigma_1 \le \sigma_1'$, $\sigma_3 \le \sigma_3'$, $\scomp{\sigma_1'}{\sigma_2}$, and $\scomp{\sigma_2}{\sigma_3'}$.
\begin{eqnarray*}
(a \otimes_{\sigma_1,\sigma_2} b) \otimes_{\sigma_1\sigma_2,\sigma_3} c 
& = & (\ext{\sigma_1,\sigma_1'}{a} \odot_{\sigma_1',\sigma_2} b)  
\otimes_{\sigma_1\sigma_2,\sigma_3} c  \\
& = & (\ext{\sigma_1,\sigma_1'}{a} \odot_{\sigma_1',\sigma_2} b)  \,\odot_{\sigma_1'\sigma_2,\sigma_3'} \ext{\sigma_3,\sigma_3'}{c}  \\
& = & \ext{\sigma_1,\sigma_1'}{a} \odot_{\sigma_1',\sigma_2\sigma_3'} (b \,\odot_{\sigma_2,\sigma_3'} \ext{\sigma_3,\sigma_3'}{c})  \\
& = & a \otimes_{\sigma_1,\sigma_2\sigma_3} (b \otimes_{\sigma_2,\sigma_3} c) 
\end{eqnarray*}

\item[Case] $\sigma_1 \le \sigma_1'$, $\sigma_2 \le \sigma_2'$, $\scomp{\sigma_1'}{\sigma_2}$, and $\scomp{\sigma_2'}{\sigma_3}$.
We have $\sigma_1' \le \sigma_1''$ and $\scomp{\sigma_1''}{\sigma_2'}$ for some $\sigma_1''$.
\begin{eqnarray*}
(a \otimes_{\sigma_1,\sigma_2} b) \otimes_{\sigma_1\sigma_2,\sigma_3} c 
& = & (\ext{\sigma_1,\sigma_1'}{a} \odot_{\sigma_1',\sigma_2} b)  
\otimes_{\sigma_1\sigma_2,\sigma_3} c  \\
& = & \ext{\sigma_1'\sigma_2,\sigma_1''\sigma_2'}{\ext{\sigma_1,\sigma_1'}{a} \odot_{\sigma_1',\sigma_2} b}  \odot_{\sigma_1''\sigma_2',\sigma_3} c  \\
& = & (\ext{\sigma_1,\sigma_1''}{a} \odot_{\sigma_1'',\sigma_2'} \ext{\sigma_2,\sigma_2'}{b})  \odot_{\sigma_1''\sigma_2',\sigma_3} c  \\
& = & \ext{\sigma_1,\sigma_1''}{a} \odot_{\sigma_1'',\sigma_2'\sigma_3} (\ext{\sigma_2,\sigma_2'}{b}  \odot_{\sigma_2',\sigma_3} c)  \\
& = & a \otimes_{\sigma_1,\sigma_2\sigma_3} (b \otimes_{\sigma_2,\sigma_3} c) 
\end{eqnarray*}

\item[Case] $\sigma_2 \le \sigma_2' \le \sigma_2''$, 
$\scomp{\sigma_1}{\sigma_2'}$, and $\scomp{\sigma_2''}{\sigma_3}$.
We have $\sigma_1 \le \sigma_1''$ and $\scomp{\sigma_1''}{\sigma_2''}$ for some $\sigma_1''$.
\begin{eqnarray*}
(a \otimes_{\sigma_1,\sigma_2} b) \otimes_{\sigma_1\sigma_2,\sigma_3} c 
& = & (a \, \odot_{\sigma_1,\sigma_2'} \ext{\sigma_2,\sigma_2'}{b})  
\otimes_{\sigma_1\sigma_2,\sigma_3} c  \\
& = & \ext{\sigma_1\sigma_2',\sigma_1''\sigma_2''}{a \, \odot_{\sigma_1',\sigma_2} \ext{\sigma_2,\sigma_2'}{b}}  \odot_{\sigma_1''\sigma_2'',\sigma_3} c  \\
& = & (\ext{\sigma_1,\sigma_1''}{a} \odot_{\sigma_1'',\sigma_2''} \ext{\sigma_2,\sigma_2''}{b})  \odot_{\sigma_1''\sigma_2'',\sigma_3} c  \\
& = & \ext{\sigma_1,\sigma_1''}{a} \odot_{\sigma_1'',\sigma_2''\sigma_3} (\ext{\sigma_2,\sigma_2''}{b}  \odot_{\sigma_2'',\sigma_3} c)  \\
& = & a \otimes_{\sigma_1,\sigma_2\sigma_3} (b
      \otimes_{\sigma_2,\sigma_3} c)
\rlap{\hbox to 144 pt{\hfill\qEd}} 
\end{eqnarray*}
\end{description}

\begin{lem} \label{lem:stackstucture2}
If $\sigma_1 \le \sigma_1'$ and $\sigma_1'\cdot \sigma_2 \ne \top$,
then $\ext{\sigma_1 \sigma_2, \sigma_1'\sigma_2}{x \otimes_{\sigma_1,\sigma_2} y} = \ext{\sigma_1,\sigma_1'}{x} \otimes_{\sigma',\sigma_2} y$.
\end{lem}
\proof \mbox{}
\begin{description}
\item[Case] $\sigma_1 \le \sigma_1''$ and $\scomp{\sigma_1''}{\sigma_2}$.
We have $(\sigma_1' \sqcup \sigma_1'') \cdot \sigma_2 = 
\sigma_1'\cdot \sigma_2 \sqcup \sigma_1''\cdot \sigma_2
= \sigma_1'\cdot \sigma_2 \sqcup \sigma_1\cdot \sigma_2
= (\sigma_1' \sqcup \sigma_1) \cdot \sigma_2 = \sigma_1' \cdot \sigma_2$,
Then, either $\sigma_1' \le \sigma_1''$ or $\sigma_1'' \le \sigma_1'$ holds.
\begin{description}
\item[Subcase] $\sigma_1' \le \sigma_1''$. 
We have $\sigma_1\cdot\sigma_2 =  \sigma_1'\cdot\sigma_2 =  \sigma_1'' \cdot\sigma_2$. 
    \begin{eqnarray*}
      \ext{\sigma_1 \sigma_2, \sigma_1'\sigma_2}{x \otimes_{\sigma_1,\sigma_2} y} & =  &
      \ext{\sigma_1'' \sigma_2, \sigma_1'\sigma_2}{\ext{\sigma_1,\sigma_1''}{x} \odot_{\sigma_1'',\sigma_2}
 y} \\
      & = & \ext{\sigma_1,\sigma_1''}{x} \odot_{\sigma_1'',\sigma_2} y \\
      & = & \ext{\sigma_1',\sigma_1''}{\ext{\sigma_1,\sigma_1'}{x}} \odot_{\sigma_1'',\sigma_2} y \\
      & = & \ext{\sigma_1,\sigma_1'}{x} \otimes_{\sigma_1',\sigma_2} y
    \end{eqnarray*} 
\item[Subcase] $\sigma_1'' \le \sigma_1'$.
  From $\scomp{\sigma_1''}{\sigma_2}$ and $\sigma_1'' \le \sigma_1'$, 
  $\sigma_2 \le \sigma_2'$ and $\scomp{\sigma_1'}{\sigma_2'}$ for some $\sigma_2'$.
    \begin{eqnarray*}
      \ext{\sigma_1 \sigma_2, \sigma_1'\sigma_2}{x \otimes_{\sigma_1,\sigma_2} y} & =  &
      \ext{\sigma_1'' \sigma_2, \sigma_1'\sigma_2}{\ext{\sigma_1,\sigma_1''}{x} \odot_{\sigma_1'',\sigma_2} y} \\
      & = & 
      \ext{\sigma_1'' \sigma_2, \sigma_1'\sigma_2'}{\ext{\sigma_1,\sigma_1''}{x} \odot_{\sigma_1'',\sigma_2} y} \\
      & = & \ext{\sigma_1,\sigma_1'}{x} \, \odot_{\sigma_1',\sigma_2'} \ext{\sigma_2,\sigma_2'}{y} \\
      & = & \ext{\sigma_1,\sigma_1'}{x} \otimes_{\sigma_1',\sigma_2} y
    \end{eqnarray*} 
\end{description}
\item[Case] $\sigma_2 \le \sigma_2'$ and $\scomp{\sigma_1}{\sigma_2'}$.
From $\scomp{\sigma_1}{\sigma_2'}$ and $\sigma_1 \le \sigma_1'$,
$\sigma_2' \le \sigma_2''$ and $\scomp{\sigma_1'}{\sigma_2''}$ for some $\sigma_2''$.
    \begin{eqnarray*}
      \ext{\sigma_1 \sigma_2, \sigma_1'\sigma_2}{x \otimes_{\sigma_1,\sigma_2} y} & =  &
      \ext{\sigma_1 \sigma_2', \sigma_1'\sigma_2''}{x \, \odot_{\sigma_1,\sigma_2'}  \ext{\sigma_2,\sigma_2'}{y}} \\
& =  &
      \ext{\sigma_1,\sigma_1'}{x} \, \odot_{\sigma_1',\sigma_2''} \ext{\sigma_2,\sigma_2''}{y} \\
      & = & \ext{\sigma_1,\sigma_1'}{x} \otimes_{\sigma_1',\sigma_2} y
\rlap{\hbox to 139 pt{\hfill\qEd}} 
    \end{eqnarray*}
\end{description}

\noindent We present a weight structure for the indexed semiring in Example~\ref{exa:hweight}.
It is almost trivial to check that it really forms a weight structure. On the other hand,
if we directly define the indexed semiring, we have to repeat proofs similar to those of
Lemma~\ref{lem:stackstucture1} and~\ref{lem:stackstucture2}.
\begin{exa} \label{exa:heightalg}
$\mytup{\set{D_\sigma}, \set{\oplus_{\sigma}}, \set{\odot_{\sigma_1,\sigma_2}}, \set{0_\sigma}, \set{1_\sigma}, \set{\extop{\sigma,\sigma'}}}$ given by the following components forms a weight structure.
\begin{itemize}
\item $D_{w/w'} = \mathbb{N}^{{\ge} \mathrm{max}(|w|,|w'|)} \cup \set{\infty}$.
\item $a \oplus_\sigma b = \mathrm{min}(a,b)$ and $0_\sigma = \infty$.
$\tup{D_\sigma, \oplus_{\sigma}, 0_\sigma}$ is clearly a commutative monoid.
\item $a \odot_{\sigma_1,\sigma_2} b = \mathrm{max}(a, b)$. It is clearly associative and its anihilator is $\infty$. 
\item $1_{w/w} = |w|$. $1_{w/w} \odot_{w/w,w/w'} b = \mathrm{max}(|w|, b) = b$ since 
$b \in \mathbb{N}^{{\ge} \mathrm{max}(|w|,|w'|)}$.
\item $\ext{w_1/w_2,w_1w/w_2w}{a} = a+|w|$.
We only show $\ext{\sigma_1 \cdot \sigma_2,\sigma_1' \cdot \sigma_2'}{a \odot b} = \ext{\sigma_1,\sigma_1'}{a} \odot \ext{\sigma_2,\sigma_2'}{b}$.
Let $\sigma_1 = w_1/w$ and $\sigma_2 = w/w_2$. Then, $\sigma_1' = w_1w'/ww'$ and $\sigma_2' = ww'/w_2w'$ for some $w'$.
\begin{eqnarray*}
\ext{\sigma_1 \cdot \sigma_2,\sigma_1' \cdot \sigma_2'}{a \odot b} & = & \mathrm{max}(a,b) + |w'| \\
& = & \mathrm{max}(a+|w'|, b+|w'|) \\
& = & \ext{\sigma_1,\sigma_1'}{a} \odot \ext{\sigma_2,\sigma_2'}{b} 
\end{eqnarray*}

\end{itemize}
\end{exa}

\section{Applications} \label{sec:application}
We present four applications of the readability analysis of weighted pushdown 
automata over indexed semirings.
The indexed semirings used in these examples are locally bounded and thus
 our framework guarantees termination of the analyses.

\subsection{Encoding of Local Variables into Weight} \label{sec:encoding}
Suwimonteerabuth applied a semiring similar to one constructed from
an indexed semiring to encode local variables of a recursive
program into weight~\cite{SuwimonteerabuthPhd}. 
Although his implementation worked without any problem,  it is 
actually not in the standard framework of weighted pushdown systems
because the semiring is not bounded.

We show that his encoding can be formulated more naturally with
an indexed semiring. In order to simplify our presentation, we give 
an encoding of a pushdown system into a weighted pushdown system with
a singleton stack alphabet. 
Since local variables can be encoded into a stack alphabet, the same
approach can be applied for the encoding of local variables.

Let us consider a singleton stack alphabet $\Gamma' = \set{\al}$. We write $m/n$ for a stack signature $\al^m/\al^n$. 
We will construct a weight structure to translate pushdown systems over a stack
alphabet $\Gamma$.
We define a weight structure $\mathcal{W}_\Gamma = \mytup{\set{D_\sigma}, \set{\oplus_{\sigma}}, \set{\odot_{\sigma_1,\sigma_2}}, \set{0_\sigma}, \set{1_\sigma}, \set{\extop{\sigma_1,\sigma_2}}}$ as follows.
\begin{itemize}
\item $D_{m/n}$ is the set of relations between $\Gamma^m$ and $\Gamma^n$: $D_{m/n} = 2^{\Gamma^m \times \Gamma^n}$.
\item $0_{m/n} = \emptyset$ and $1_{m/m} = \set{\tup{x,x} \mid x \in \Gamma^m}$.
\item $R_1 \odot_{l/m,m/n} R_2$ is a composition of two relations $R_1$ and $R_2$: $R_1 \circ R_2$
where $R_1 \subseteq \Gamma^l \times \Gamma^m$ and
$R_2 \subseteq \Gamma^m \times \Gamma^n$.

\item $R_1 \oplus_{m/n} R_2$ is the union of two relations $R_1$ and $R_2$: 
$R_1 \cup R_2$ where
$R_1, R_2 \subseteq \Gamma^m \times \Gamma^n$.
\item $\extop{l/m,l+1/m+1}$ extends the domain of a relation and is defined by
\[ \ext{l/m, l+1/m+1}{R} = 
\set{ \tup{\tup{x, z}, \tup{y, z}}
\mid \tup{x,y} \in R \land z \in \Gamma} \]
where we consider $\Gamma^{k+1} = \Gamma^k \times \Gamma$.
\end{itemize}
It is straightforward to show this structure forms a weight structure.
Furthermore, it induces a locally bounded indexed semiring because $D_{m/n}$ is 
the power set of a finite set and ordered by the set inclusion. 

We show how to simulate a pushdown system $\mathcal{P} = \tup{P, \Gamma, \Delta}$ 
by a weighted pushdown system $\mathcal{P'}$ over
the weight structure $\mathcal{W}_\Gamma$. 
Let $\mathcal{P'}$ be $\tup{P, \Gamma', \Delta'}$ such that 
\[ \tup{p, \#, p', \#^m, a} \in \Delta' \quad\quad\mbox{iff}\quad\quad 
\tup{p, \gamma, p', w} \in \Delta \]
where $|w| = m$ and $a = \set{ \tup{\gamma,w} }$. 

Then, $\mathcal{P}$ and $\mathcal{P}'$ are equivalent in the following sense:
\[ \utransrule{p}{p'}{w/w'}{\mathcal{P}} \quad\quad \Longleftrightarrow 
\quad\quad
\atransrulex{p}{p'}{m/m'}{a}{\mathcal{P}'} \land \tup{w,w'} \in a \]
where $m = |w|$ and $m' = |w'|$.
Then, we can check the reachability in $\mathcal{P}$ by checking
that in $\mathcal{P}'$.

\subsection{Conditional Pushdown Systems\label{sec:conditional}} 
Esparza \textit{et al.}\ introduced pushdown systems with checkpoints that
have the ability to inspect the whole stack content against a regular language~\cite{esparza:regularvaluation}.
Li and Ogawa reformulated their definition and called them conditional pushdown systems~\cite{li:conditional}. 
We review conditional pushdown systems and then formulate the reachability analysis in our previous work~\cite{minamide:html5} as that of weighted pushdown systems.

\begin{defi}
A conditional pushdown system $\mathcal{P}$ is a structure
$\tup{P, \Gamma, \Delta}$ 
where $P$ is a finite set of states, 
$\Gamma$ is a stack alphabet, 
and $\Delta \subseteq P \times \Gamma \times P \times \Gamma^* 
\times \mathrm{Reg}(\Gamma)$
is a set of transitions where
$\mathrm{Reg}(\Gamma)$ is the set of regular languages over $\Gamma$. 
\end{defi}

We write $\ptrans{p}{\gamma}{p'}{w}{R}$ if $\tup{p, \gamma, p', w, R} \in  \Delta$ as weighted pushdown systems.
The transition relation of a conditional pushdown system is defined as follows.
\begin{itemize}
\item $\utransrule{\tup{p,w}}{\tup{p,w}}{}{}$.
\item $\utransrule{\tup{p,\gamma w'}}{\tup{p',ww'}}{}{}$ if $\ptrans{p}{\gamma}{p'}{w}{R}$ and $w' \in R$.
\item $\utransrule{\tup{p,w}}{\tup{p',w'}}{}{}$
if $\utransrule{\tup{p,w}}{\tup{p'',w''}}{}{}$ and
$\utransrule{\tup{p'',w''}}{\tup{p',w'}}{}{}$.
\end{itemize}
In the second case above, the transition can be taken only when
the current stack content excluding its top is included in 
the regular language $R$ given as the condition of the rule.

We show that the transition of a conditional pushdown system can be simulated
by that of a weighted pushdown system without conditional rules.
Let us design a weight structure for this simulation.
We use the same domain  for all proper stack signatures $\sigma$: $D_\sigma = 2^{\Gamma^*}$. 
Then, the weight structure $\mytup{\set{D_\sigma}, \set{\oplus_{\sigma}}, \set{\odot_{\sigma_1,\sigma_2}}, \set{0_\sigma}, \set{1_\sigma}, \set{\extop{\sigma,\sigma'}}}$ is given as follows.
\begin{itemize}
\item $0_\sigma = \emptyset$ and $1_\sigma = \Gamma^*$.
\item $a \oplus_\sigma b = a \cup b$.
\item $a \odot_{\sigma_1,\sigma_2} b = a \cap b$ for strictly compatible signatures $\sigma_1$ and $\sigma_2$.
\item $\ext{w_1/w_2,w_1w/w_2w}{a} =  w^{-1}a$
where $w^{-1}a$ is left quotient defined by $w^{-1}a = \set{w' \mid w w' \in a}$.
\end{itemize}
From basic properties
of left quotient and set operations, it is clear that this structure forms a weight structure.
Then, for a conditional pushdown system $\mathcal{P}$ we obtain a weighted 
pushdown system $\mathcal{P}'$ over the indexed semiring above
by considering a conditional transition rule $\ptrans{p}{\gamma}{p'}{w}{R}$
as a weighted one.

A conditional pushdown system $\mathcal{P}$ is simulated by
a weighted pushdown system $\mathcal{P}'$ in the following sense.
\begin{itemize}
\item If $\utransrule{\tup{p_1, w_1}}{\tup{p_2, w_2}}{}{\mathcal{P}}$,
then there exist $w$, $w_1'$, and $w_2'$ such that 
$\atransrulex{p_1}{p_2}{w_1'/w_2'}{a}{\mathcal{P}'}$,
$w \in a$, and $w_1/w_2 = w_1'w/w_2'w$.
\item If $\atransrulex{p_1}{p_2}{w_1/w_2}{a}{\mathcal{P}'}$ and
$w \in a$,
then $\utransrule{\tup{p_1, w_1w}}{\tup{p_2, w_2w}}{}{\mathcal{P}}$.
\end{itemize}

Please note that this weight structure is not locally bounded because
$2^{\Gamma^*}$ is not bounded with respect to the set inclusion.
However, $D_\sigma$ can be restricted to the set $D \subseteq 2^{\Gamma^*}$ inductively defined as 
follows.
\begin{itemize}
\item $\emptyset \in D$ and $\Gamma^* \in D$.
\item $R \in D$ if $\ptrans{p}{\gamma}{p'}{w}{R}$ for some $p$, $\gamma$, $p'$, $w$.
\item $R_1 \cap R_2 \in D$ and $R_1 \cup R_2 \in D$ if $R_1 \in D$ and $R_2 \in D$.
\item $w^{-1}R \in D$ if $R \in D$ and $w \in \Gamma^*$.
\end{itemize}
This set $D$ is finite because the set of transitions is finite, 
there are finitely many languages obtained from each regular language 
with left quotient, and left quotient distributes
over union and intersection.
Thus, we obtain a locally bounded indexed semiring by using $D$.
This gives the algorithm of the backward reachability analysis for
conditional pushdown systems that we used to analyze the HTML5
parser specification~\cite{minamide:html5}.

\subsection{Well-Structured Pushdown Systems}
Cai and Ogawa introduced well-structured pushdown systems (WSPDS) where 
the set of states and stack alphabet 
can be possibly infinite well-quasi-ordered sets.
They showed that the coverability problem is decidable for
WSPDS with a finite set of states and then extended the result for
several subclasses of WSPDS~\cite{ogawa:well}.
We show that the coverability of WSPDS with a finite set of
states can also be decided through a translation to weighted 
pushdown systems with indexed weight domains.

A quasi-ordering $(D, {\preceq})$ is a reflexive and transitive binary 
relation on $D$. A quasi-order $(D, {\preceq})$ is a well-quasi-order  
if, for each infinite sequence $a_1, a_2, a_3, \ldots$ in $D$,
there exist $i, j$ such that $i < j$ and $a_i \preceq a_j$.
A set $I \subseteq D$ is an ideal if $a \in I$ and $a \preceq b$ imply $b \in I$.
The upward closure of $A \subseteq D$ is 
$\upcl{A} = \set{ b \in D \mid \exists a \in A. a \preceq b}$.
The family of ideals over $A$ is denoted by $\idealof{A}$.

Well-structured pushdown systems are defined as follows where
$\PFUN(A,B)$ denotes the set of partial functions from $A$ to $B$.
\begin{defi}
A well-structured pushdown system is a structure $\tup{P, \Gamma, \Delta}$
where $P$ is a finite set of states, $\Gamma$ is a possibly infinite
set of stack symbols with well-quasi-order $\preceq$, and 
$\Delta \subseteq P \times P \times 
\bigcup_{i \in \mathbb{N}} \PFUN(\Gamma,\Gamma^{i})$ is a finite set
of monotonic transition rules.
A transition rule $\tup{p, p', \phi}$ is monotonic if $\phi$ is monotonic on $\preceq$.
\end{defi}
If $\tup{p, p', \phi} \in \Delta$ and $\phi \in \PFUN(\Gamma,\Gamma^{i})$, 
then $\phi^{-1}(X) \in \idealof{\Gamma}$ for any $X \in \idealof{\Gamma^i}$
by the monotonicity of $\phi$.
The transition relation of a WSPDS is defined as follows.
\begin{itemize}
\item $\utransrule{\tup{p,w}}{\tup{p,w}}{}{}$.
\item $\utransrule{\tup{p,\gamma w'}}{\tup{p',\phi(\gamma)w'}}{}{}$ if $\tup{p,p',\phi} \in \Delta$ and $\phi(\gamma)$ is defined.
\item $\utransrule{\tup{p,w}}{\tup{p',w'}}{}{}$
if $\utransrule{\tup{p,w}}{\tup{p'',w''}}{}{}$ and
$\utransrule{\tup{p'',w''}}{\tup{p',w'}}{}{}$.
\end{itemize}

Cai and Ogawa showed that the coverability problem of WSPDS is decidable.
We say that $\tup{p_2,w_2}$ is \emph{covered} by $\tup{p_1,w_1}$ if we have
$\utransrule{\tup{p_1,w_1}}{\tup{p_2,w_2'}}{}{}$ for some $w_2'$ such that $w_2 \preceq w_2'$.
The key to the development of the coverability analysis of WSPDS by Cai and Ogawa is the following lemma. This also makes it possible to construct a locally bounded indexed semiring.
\begin{lem}[Finkel et al.~\cite{finkel98b}] \label{lem:ideal}
If $\preceq$ is a well-quasi-order, then any infinite sequence $I_0 \subseteq I_1 \subseteq I_2 \subseteq \cdots$ of ideals eventually stabilizes.
\end{lem}

For the coverability analysis, we translate a WSPDS into a weighted pushdown system 
with a singleton stack alphabet $\Gamma' = \set{\al}$. Then we translate
the transition rule $\tup{p, p', \phi} \in \Delta$ in WSPDS into the following transition in
a weighted pushdown system $\mathcal{P}'$:
\[\utransrule{\tup{p,\al}}{\tup{p',\al^i}}{\phi^{-1}}{\mathcal{P}'} \]
where $\phi \in \PFUN(\Gamma,\Gamma^{i})$. We adopt $\phi^{-1}$ as a weight instead of
$\phi$ because we apply $\phi^{-1}(X) \in \idealof{\Gamma}$ for any $X \in \idealof{\Gamma^i}$.
The weight structure $\mytup{\set{D_\sigma}, \set{\oplus_{\sigma}}, \set{\odot_{\sigma_1,\sigma_2}}, \set{0_\sigma}, \set{1_\sigma}, \set{\extop{\sigma_1,\sigma_2}}}$ is defined as follows.
\begin{itemize}
\item $D_{m/n} =\Gamma^n \to \idealof{\Gamma^m}$.
\item $0_{m/n} = \lambda x. \emptyset$ and $1_{m/m} = \lambda x.\upcl{\set{x}}$.
\item $f_1 \odot_{l/m,m/n} f_2$ is the composition of functions: $\hat{f_1} \circ f_2$ where $\hat{f_1}(X) = \bigcup_{x\in X}f_1(x)$.
\item $f_1 \oplus_{m/n} f_2$ is defined by $\lambda x. f_1(x) \cup f_2(x)$.
\item $\extop{l/m,l+1/m+1}$ extends the domain and range of a function and 
is defined as follows:
\[ \ext{l/m,l+1/m+1}{f} = \lambda \tup{y,z}. f(y) \times \upcl{\set{z}} \]
where $y \in \Gamma^m$ and $z \in \Gamma$.
\end{itemize}
$\tup{D_{m/n},\oplus_{m/n}, 0_{m/n}}$ is clearly a commutative monoid.
The other properties of a weight structure can be easily verified.
Furthermore, it induces a locally bounded indexed semiring because $D_{m/0}$ is 
isomorphic to $\idealof{\Gamma^m}$ and there are no infinite 
ascending chains of ideals by Lemma~\ref{lem:ideal}.
It should be noted that $D_{m/n}$ is not bounded in general for $n > 0$.

We translate a WSPDS $\mathcal{P} = \tup{P, \Gamma, \Delta}$ to
a weighted pushdown system $\mathcal{P'} = \tup{P, \Gamma', \Delta'}$
 over the above weight structure.
The set of transition rules $\Delta'$ is defined by 
\[ \tup{p, \#, p', \#^i, a} \in \Delta' \quad\quad\mbox{if}\quad\quad 
\tup{p, p', \phi} \in \Delta \;\mbox{and} \; \phi \in \PFUN(\Gamma,\Gamma^i) \]
where $a = \lambda w. \phi^{-1}(\upcl{\set{w}})$.

Then, $\mathcal{P}$ and $\mathcal{P}'$ are closely related in the following sense.
The proof appears in Appendix~\ref{appendix:wellpds}.
\begin{prop} \label{prop:wellpds} \mbox{}

\begin{itemize}
\item If $\utransrule{\tup{p_1,w_1}}{\tup{p_2,w_2}}{}{\mathcal{P}}$,
then $\utransrule{\tup{p_1,m_1}}{\tup{p_2,m_2}}{a}{\mathcal{P}'}$ and $w_1 \in  a(w_2)$.
\item If $\utransrule{\tup{p_1,m_1}}{\tup{p_2,m_2}}{a}{\mathcal{P}'}$ and 
$w_1 \in  a(w_2)$, then 
$\utransrule{\tup{p_1,w_1}}{\tup{p_2,w_2'}}{}{\mathcal{P}}$ for some $w_2 \preceq w_2'$.
\end{itemize}
 where $m_1 = |w_1|$ and $m_2 = |w_2|$.
\end{prop}

Then, the coverability in $\mathcal{P}$ can be checked by applying 
the reachability analysis to $\mathcal{P}'$ in the following manner.
Let us consider the coverability of $\tup{p,w}$ for $w=\gamma_1\gamma_2\cdots\gamma_n$.
We represent $w$ 
by a weighted automaton $\mathcal{A}_w = \tup{\set{q_0, q_1, \ldots, q_n}, \set{\al}, \Delta_w, q_0, \set{q_n}}$
where $\tup{q_{i-1},q_{i}, \al, \upcl{\set{\gamma_i}}} \in \Delta_w$ for $1 \le i \le n$.
Then, $\tup{p,w}$ is covered by $\tup{p',w'}$ in $\mathcal{P}$  if and only if
$w' \in \delta_{\mathcal{P}',\mathcal{A}_w}(p,\al^m,p')$ where $m = |w'|$.

\subsection{Pushdown Systems with Stack Manipulation} \label{sec:trpds}
Uezato and Minamide introduced pushdown systems with stack manipulation (TrPDS) 
that can modify the whole stack content with a \emph{letter-to-letter} finite-state 
transducer at each transition~\cite{uezato:stack}.
TrPDS generalizes conditional pushdown systems~\cite{esparza:regularvaluation,li:conditional} and discrete timed pushdown 
systems~\cite{Abdulla:discrete}.
They showed that the reachability problem of a TrPDS is decidable if the closure of
transductions appearing in the transition rules is finite.

The behaviour of a letter-to-letter transducer whose input and output alphabets  
are $\Gamma$ is characterized by a regular language over $\Gamma \times \Gamma$.
Thus, we identify a letter-to-letter transducer with a corresponding 
regular language over $\Gamma \times \Gamma$ and call it a \emph{transduction}.
Let $w = a_1a_2\cdots a_n$ and $w' = b_1b_2\cdots b_n$.
We abuse the tuple notation and write $\tup{w, w'}$ for $\tup{a_1,b_1}\tup{a_1,b_1}\cdots\tup{a_n,b_n}$ if it is clear from the context.
For a transduction $t$, the left quotient of the transduction is defined as follows:
$\tup{\gamma_1, \gamma_2}^{-1}t = \set{ \tup{w_1, w_2} \mid \tup{\gamma_1w_1, \gamma_2w_2} \in t}$.

We say that $\mathcal{T} \subseteq \mathrm{Reg}(\Gamma \times \Gamma)$ is \emph{closed} if
the following hold.
\begin{itemize}
\item $\emptyset \in \mathcal{T}$ and $\set{\tup{w,w} \mid w \in \Gamma^*} \in \mathcal{T}$.
\item If $t_1, t_2 \in \mathcal{T}$, then $t_1 \circ t_2 \in \mathcal{T}$ and
$t_1 \cup t_2 \in \mathcal{T}$.
\item If $t \in \mathcal{T}$, then $\tup{\gamma_1,\gamma_2}^{-1}t \in \mathcal{T}$ for all $\gamma_1,\gamma_2 \in \Gamma$.
\end{itemize}
We sometimes write $0_\mathcal{T}$ and $1_\mathcal{T}$ for $\emptyset$ and $\set{\tup{w,w} \mid w \in \Gamma^*}$, respectively.

\begin{defi}
A TrPDS $\mathcal{P}$ is a structure $\tup{P, \Gamma, \mathcal{T}, \Delta}$ 
where $P$ is a finite set of states, 
$\Gamma$ is a stack alphabet, 
$\mathcal{T} \subseteq \mathrm{Reg}(\Gamma \times \Gamma)$ is a finite, closed set of
transductions, 
and $\Delta \subseteq P \times \Gamma \times P \times \Gamma^* 
\times \mathcal{T}$
is a set of transitions.
\end{defi}

We write $\ptrans{p}{\gamma}{p'}{w}{t}$ if $\tup{p, \gamma, p', w, t} \in  \Delta$ as weighted pushdown systems.
The transition relation of a TrPDS is defined as follows.
\begin{itemize}
\item $\utransrule{\tup{p,w}}{\tup{p,w}}{}{}$.
\item $\utransrule{\tup{p,\gamma w'}}{\tup{p',ww''}}{}{}$ if $\ptrans{p}{\gamma}{p'}{w}{t}$ and $\tup{w',w''} \in t$.
\item $\utransrule{\tup{p,w}}{\tup{p',w'}}{}{}$
if $\utransrule{\tup{p,w}}{\tup{p'',w''}}{}{}$ and
$\utransrule{\tup{p'',w''}}{\tup{p',w'}}{}{}$.
\end{itemize}
In the second case above, the stack content below the top is
modified by the transduction $t$.

A TrPDS can be simulated by combining the ideas of simulations in
Section~\ref{sec:encoding} and \ref{sec:conditional}.
We again use the singleton stack alphabet $\Gamma' = \set{\al}$
and define weight structure $\mytup{\set{D_\sigma}, \set{\oplus_{\sigma}}, \set{\odot_{\sigma_1,\sigma_2}}, \set{0_\sigma}, \set{1_\sigma}, \set{\extop{\sigma_1,\sigma_2}}}$ as follows.
\begin{itemize}
\item $D_{m/n} = \Gamma^m \times \Gamma^n \to \mathcal{T}$.
\item $0_{m/n}(w_1,w_2) = 0_{\mathcal{T}}$ and \label{fix:trpds}
\[ \begin{array}{rcl}
1_{m/m}(w_1,w_2) & =  &
\left\{ \begin{array}{ll} 
    1_{\mathcal{T}} & \mbox{(if $w_1 = w_2$)} \\
    0_{\mathcal{T}} & \mbox{(otherwise)}.
                 \end{array} \right.
             \end{array}
             \]
                 
\item For
$f_1 \in \Gamma^l \times \Gamma^m \to \mathcal{T}$ and
$f_2 \in \Gamma^m \times \Gamma^n \to \mathcal{T}$, 
$f_1 \odot_{l/m,m/n} f_2$ is defined by
\[ \lambda (w_1,w_3). \bigcup_{w_2 \in \Gamma^m} f_1(w_1,w_2) \circ f_2(w_2,w_3). \]

\item For $f_1, f_2 \in \Gamma^m \times \Gamma^n \to \mathcal{T}$,
$f_1 \oplus_{m/n} f_2$ is defined by
\[ \lambda (w_1,w_2). f_1(w_1,w_2) \cup f_2(w_1,w_2). \]
\item $\extop{l/m,l+1/m+1}$ extends the domain of a function and is defined 
by
\[ \ext{l/m,l+1/m+1}{f}(w_1\gamma_1,w_2\gamma_2) = \tup{\gamma_1,\gamma_2}^{-1}f(w_1,w_2). \]
\end{itemize}
This structure forms a weight structure, and induces a locally bounded indexed semiring because $\mathcal{T}$ is a finite set.

We simulate a TrPDS $\mathcal{P}=\tup{P, \Gamma, \Delta}$ by a weighted pushdown system
 $\mathcal{P'} = \tup{P, \set{\al}, \Delta'}$. 
For a transduction $t \in \mathcal{T}$, we define the function 
$\trstep{t}{\gamma}{w} : \Gamma \times \Gamma^{|w|} \to \mathcal{T}$ as follows.
\[
\trstep{t}{\gamma}{w}(\gamma', w') = 
\left\{
\begin{array}{ll}
t & \mbox{if $\gamma' = \gamma$ and $w' = w$} \\
0_\mathcal{T} & \mbox{otherwise}
\end{array}
\right.
\]
Then, $\Delta'$ is given by 
\[ \tup{p, \#, p', \#^{|w|}, \trstep{t}{\gamma}{w}} \in \Delta' \quad\quad\mbox{iff}\quad\quad \tup{p, \gamma, p', w, t} \in \Delta. \]
$\mathcal{P}$ is simulated by $\mathcal{P}'$ in the following sense.
Hence, the reachability in $\mathcal{P}$ can be decided by the reachability analysis in $\mathcal{P}'$. 
The proof of the following proposition appears in Appendix~\ref{appendix:trpds}.
\begin{prop} \label{prop:trpds}
Let $m_1 = |w_1|$ and $m_2 = |w_2|$.
\begin{itemize}
\item If $\utransrule{\tup{p_1,w_1}}{\tup{p_2,w_2}}{}{\mathcal{P}}$,
then $\utransrule{\tup{p_1,m_1}}{\tup{p_2,m_2}}{a}{\mathcal{P}'}$ and 
$\tup{\epsilon,\epsilon} \in a(w_1,w_2)$ for some $a$.

\item If $\utransrule{\tup{p_1,m_1}}{\tup{p_2,m_2}}{a}{\mathcal{P}'}$ and 
$\tup{\epsilon,\epsilon} \in a(w_1,w_2)$, then 
$\utransrule{\tup{p_1,w_1}}{\tup{p_2,w_2}}{}{\mathcal{P}}$.
\end{itemize}
\end{prop}

The backward reachability analysis similar to the above was presented by
Uezato and Minamide~\cite{uezato:stack}. 
However, they used an ad-hoc extension of automata to generalize the saturation 
procedure and their presentation was rather complicated. 
We here greatly clarify the presentation by using our framework of weighted 
pushdown systems.

\section{Related Work}
An automaton over a monoid $\mathcal{M}$ is called a generalized $\mathcal{M}$-automaton by
Eilenberg~\cite{eilenberg:automata}. 
The textbook of Sakarovitch discusses automata over several classes of
monoids including free groups and commutative monoids~\cite{sakarovitch:elements}.
As far as we know, this paper is the first work that discusses the reachability
analysis of pushdown systems by considering them as automata over the monoid of stack signatures.

Let us consider a paired alphabet $\widetilde{\Gamma} = \Gamma \cup \overline{\Gamma}$
where $\overline{\Gamma}  = \set{\overline{a} \mid a \in \Gamma}$.
Letters $\gamma$ and $\overline{\gamma}$ correspond to a push and a pop of 
$\gamma$, respectively. 
Then, the monoid $\mathcal{M}_\Gamma$ is closely related to the monoid over
$\widetilde{\Gamma}^*$ obtained by Shamir congruence~\cite{shamir:representation}, which is generated by $\gamma\overline{\gamma} = \epsilon$. If we add the relation
$\gamma\overline{\gamma'} = \top$ for $\gamma \ne \gamma'$, then the reduced form of a word over
$\widetilde{\Gamma}$ has the following form: $\overline{w_1}w_2$ or $\top$.
If we write $w_1/{w_2}^R$ for $\overline{w_1}w_2$, we obtain a stack signature\footnote{${w_2}^R$ is the reverse of $w_2$.}.

Esparza \textit{et al.} showed that
conditional pushdown systems can be translated to ordinary pushdown systems~\cite{esparza:regularvaluation}. 
Hence, the reachability can be decided via the translation. 
However, it is not practical to apply the translation because of exponential blowup of the size of pushdown systems. 
The algorithm formulated in Section~\ref{sec:conditional} as the reachability analysis of
weighted pushdown systems has also an exponential complexity.
However, it avoids the 
exponential blowup by the translation before applying the reachability analysis and worked well for the analysis of the HTML5 parser specification~\cite{minamide:html5}.

Reps \textit{et al.}~\cite{reps:weighted} developed both of the
forward and backward analysis of weighted pushdown systems. 
Although our backward analysis is a direct extension of their analysis,
the forward reachability analysis cannot directly be extended for
indexed weight domains. This is because $a \in D_{\gamma / \gamma'\gamma''}$ 
cannot be decomposed to $a = a_1 \otimes a_2$ for
$a_1 \in D_{\gamma /\gamma''}$  and $a_2 \in D_{\epsilon/\gamma'}$ in general.
If this decomposition is possible, a slightly modified version of
their forward reachability analysis can be extended for indexed weighted 
domains (we add a new states $q_r$ 
indexed by a transition rule $r$  as the original forward reachability analysis considered by Esparza~\textit{et.al}~\cite{esparza:efficient} instead of $q_{p',\gamma'}$ indexed by a state $p'$ and a pushdown symbol $\gamma'$.). 
However, among the four indexed semirings in Section~\ref{sec:application}, 
only the indexed semiring for conditional pushdown systems enables the 
decomposition above.
 It should be noted that Cai and Ogawa developed the forward reachability analysis of well-structured pushdown systems by combining the saturation procedure
with the Karp-Miller acceleration instead of the ideal representation~\cite{ogawa:well}.

\section{Conclusions}
We have introduced the monoid of stack signatures to treat
pushdown systems as automata over the monoid. 
Then, weighted pushdown systems are generalized by adopting
a semiring indexed by stack signatures as weight. 
This generalization makes it possible to relax the restriction
of boundedness and extend the applications of the reachability analysis 
of weighted pushdown systems.

We have shown that by designing proper indexed semirings, the reachability analysis     
of several extensions of pushdown systems can be achieved by a translation to
weighted pushdown systems and their reachability analysis.
Although the reachability analysis of those extensions were already developed 
by directly extending the analysis of ordinary pushdown systems, our approach
clarifies the analysis by separating  the design of indexed semirings,
which depends on each extension, from the general algorithm of the 
reachability analysis.

The indexed semirings for the applications in this paper are given
through weight structures. We consider that it is simpler to construct 
and implement indexed semirings through weight structures than to directly construct them. However, we are not completely satisfied with the formulation
of weight structures because their definition looks rather ad-hoc mathematically.
We would like to investigate more abstract notion corresponding to
weight structures.

\section*{Acknowledgments}
I would like to thank Stefan Schwoon for inspiring discussions and
suggestions. Schwoon informed me of the work of Suwimonteerabuth on
the encoding of local variables into weight.
I would also like to thank Shin-ya Katsumata for his comments on
lax monoidal functors and graded semirings.
The paper has also benefited from constructive feedback and suggestions
by the anonymous referees, which are greatly appreciated.
This work has been partially supported by JSPS Grant-in-Aid for Science
Research (C) 24500028 and 15K00087, and the Kayamori Foundation of
Informational Science Advancement.

\bibliographystyle{alpha}
\bibliography{sa}

\appendix
\section{Proofs on Stack Signatures} \label{appendix:assoc}

\begin{lem}
$(w_1/w_1' \cdot w_2/w_2') \cdot w_3/w_3' = 
 w_1/w_1' \cdot (w_2/w_2' \cdot w_3/w_3')$
\end{lem}
\proof By case analysis on the prefix relation. We omit the cases
where $(w_1/w_1' \cdot w_2/w_2') \cdot w_3/w_3' = 
 w_1/w_1' \cdot (w_2/w_2' \cdot w_3/w_3') = \top$.
\begin{enumerate}
\item $w_1'$ is a prefix of $w_2$, i.e., $w_2 = w_1'w_2''$.
\begin{enumerate}
\item $w_2'$ is a prefix of $w_3$, i.e., $w_3 = w_2'w_3''$.
\begin{eqnarray*}
(w_1/w_1' \cdot w_2/w_2') \cdot w_3/w_3' 
         & = & w_1w_2''/w_2' \cdot w_3/w_3' \\
         & = & w_1w_2''w_3''/w_3' \\
         & = & w_1/w_1' \cdot w_1'w_2''w_3''/w_3'  \\
         & = & w_1/w_1' \cdot (w_2/w_2' \cdot w_3/w_3')  
\end{eqnarray*}
\item $w_3$ is a prefix of $w_2'$, i.e., $w_2' = w_3w_2'''$.       
\begin{eqnarray*}
(w_1/w_1' \cdot w_2/w_2') \cdot w_3/w_3'  
         & = & w_1w_2''/w_2' \cdot w_3/w_3' \\
         & = & w_1w_2''/w_3'w_2''' \\
         & = & w_1/w_1' \cdot w_2/w_3'w_2'''  \\
         & = & w_1/w_1' \cdot (w_2/w_2' \cdot w_3/w_3')  
\end{eqnarray*}
\end{enumerate}
\item $w_2$ is a prefix of $w_1'$, i.e., $w_1' = w_2w_1''$.
\begin{enumerate}
\item $w_2'$ is a prefix of $w_3$, i.e., $w_3 = w_2'w_3''$.
\begin{enumerate}
\item $w_1''$ is a prefix of $w_3''$, i.e., $w_3'' = w_1''w$.
\begin{eqnarray*}
 (w_1/w_1' \cdot w_2/w_2') \cdot w_3/w_3' 
         & = & w_1/w_2'w_1'' \cdot w_2'w_3''/w_3' \\
         & = & w_1w/w_3' \\
         & = & w_1/w_2w_1'' \cdot w_2w_3''/w_3'  \\
         & = & w_1/w_1' \cdot (w_2/w_2' \cdot w_3/w_3')  
\end{eqnarray*}
\item $w_3''$ is a prefix of $w_1''$. Symmetric to the case above.
\end{enumerate}
\item $w_3$ is a prefix of $w_2'$, i.e., $w_2' = w_3w_2'''$.  
This case is symmetric to Case (1a).\qed
\end{enumerate}
\end{enumerate}

\begin{lem} \label{lem:cases}
If $\sigma_1 \cdot \sigma_2 \cdot \sigma_3 \ne \top$, one of the followings holds.
\begin{enumerate} 
\item $\sigma_1 \le \sigma_1'$, $\sigma_3 \le \sigma_3'$, $\scomp{\sigma_1'}{\sigma_2}$, and $\scomp{\sigma_2}{\sigma_3'}$.

\item $\sigma_1 \le \sigma_1'$, $\sigma_2 \le \sigma_2'$, $\scomp{\sigma_1'}{\sigma_2}$, and $\scomp{\sigma_2'}{\sigma_3}$.

\item $\sigma_3 \le \sigma_3'$, $\sigma_2 \le \sigma_2'$, 
$\scomp{\sigma_2}{\sigma_3'}$, and $\scomp{\sigma_1}{\sigma_2'}$.

\item $\sigma_2 \le \sigma_2' \le \sigma_2''$, 
$\scomp{\sigma_1}{\sigma_2'}$, and $\scomp{\sigma_2''}{\sigma_3}$.

\item $\sigma_2 \le \sigma_2' \le \sigma_2''$, 
$\scomp{\sigma_1}{\sigma_2''}$, and $\scomp{\sigma_2'}{\sigma_3}$.
\end{enumerate}
\end{lem}
\proof
This lemma is obtained by inspecting the proof of the above lemma.
\qed

\begin{lem} \label{lem:ordered}
If $\sigma_1 \le \sigma_1'$ and $\sigma_2 \le \sigma_2'$, then $\sigma_1 \cdot \sigma_2 \le \sigma_1' \cdot \sigma_2'$.
\end{lem}
\proof
It is sufficient to prove the proposition for the case $\sigma_1' \cdot \sigma_2' \ne \top$.
Then, there exist strictly compatible $\sigma_1''$ and $\sigma_2''$ such that
$\sigma_1' \le \sigma_1''$, $\sigma_2' \le \sigma_2''$, and
$\sigma_1' \cdot \sigma_2' = \sigma_1'' \cdot \sigma_2''$.
Thus, we can assume that $\sigma_1'$ and $\sigma_2'$ are strictly compatible.
\begin{description}
\item[Case $\sigma_1 \cdot \sigma_2 \ne \top$]
Without loss of generality, we assume that $\sigma_1 = w_1/w$ and
$\sigma_2 = w w_2/w_2'$. Then, we have $\sigma_1' = w_1w_2w'/w w_2 w'$ and
$\sigma_2' = ww_2w'/w_2'w'$ for some $w'$.
Hence, $w_1 w_2/w_2' = \sigma_1 \cdot \sigma_2 \le 
\sigma_1' \cdot \sigma_2' = w_1 w_2w'/w_2'w'$.
\item[Case $\sigma_1 \cdot \sigma_2 = \top$]
This case contradicts $\sigma_1' \cdot \sigma_2' \ne \top$.
\end{description}
\qed

\begin{lem} \label{lem:lele}
Let $\sigma \ne \top$.
If $\sigma_1 \le \sigma$ and $\sigma_2 \le \sigma$, then either
$\sigma_1 \le \sigma_2$ or $\sigma_2 \le \sigma_1$.
\end{lem}
\proof
This lemma can be easily proved by  case analysis.
\qed

\begin{lem} \label{lem:distrib}
$(\sigma_1 \sqcup \sigma_2) \cdot \sigma_3 =
(\sigma_1 \cdot \sigma_3) \sqcup (\sigma_2 \cdot \sigma_3)$.
\end{lem}
\proof
If $\sigma_1 \le \sigma_2$, then
$\sigma_1 \cdot \sigma_3 \le \sigma_2 \cdot \sigma_3$ by Lemma~\ref{lem:ordered} and thus the proposition holds.
To cover the other case, 
we show $\sigma_1 \sqcup \sigma_2 \ne \top$
by assuming $(\sigma_1 \cdot \sigma_3) \sqcup (\sigma_2 \cdot \sigma_3) \ne \top$.
\begin{description}
\item[Case 1] $\sigma_1\cdot\sigma_3 = \sigma_1\cdot\sigma_3'$ 
for strictly compatible $\sigma_1$ and $\sigma_3'$, 
and $\sigma_2\cdot\sigma_3 = \sigma_2\cdot\sigma_3''$
for strictly compatible $\sigma_2$ and $\sigma_3''$.
By Lemma~\ref{lem:lele}, without loss of generality, we assume $\sigma_1\cdot\sigma_3 \le \sigma_2\cdot\sigma_3$.

Let $\sigma_1 = w_1/w_1'$, $\sigma_2 = w_2/w_2'$, and $\sigma_3 = w_3/w_3'$.
Then, $w_1'  = w_3 w_{13}$ and $w_2'  = w_3 w_{23}$ for some
$w_{13}$ and $w_{23}$.
Then, $\sigma_1\cdot\sigma_3 = w_1/w_3'w_{13}$ and 
$\sigma_2\cdot\sigma_3 = w_2/w_3'w_{23}$.
From $\sigma_1\cdot\sigma_3 \le \sigma_2\cdot\sigma_3$,
$w_2 = w_1w$ and $w_{23} = w_{13}w$ for some $w$.
Then, $\sigma_1 = w_1/w_3 w_{13}$ and $\sigma_2 = w_1w/w_3 w_{13}w$.

\item[Case 2]
$\sigma_1\cdot\sigma_3 = \sigma_1\cdot\sigma_3'$ for strictly compatible
$\sigma_1$ and $\sigma_3'$, and
$\sigma_2\cdot\sigma_3 = \sigma_2'\cdot\sigma_3$
for strictly compatible $\sigma_2'$ and $\sigma_3$.
Let $\sigma_1 = w_1/w_1'$, $\sigma_2 = w_2/w_2'$, and $\sigma_3 = w_3/w_3'$.
Then, $w_1'  = w_3 w_{13}$ and $w_3  = w_2' w_{23}$ for some
$w_{13}$ and $w_{23}$.
Then, $\sigma_1\cdot\sigma_3 = w_1/w_3'w_{13}$ and 
$\sigma_2\cdot\sigma_3 = w_2w_{23}/w_3'$.
\begin{itemize}
\item Subcase $\sigma_2\cdot\sigma_3 \le \sigma_1\cdot\sigma_3$.
Then, we have $w_1 = w_2w_{23}w_{13}$ hence
$\sigma_1 = w_2w_{23}w_{13}/w_2'w_{23}w_{13}$ and therefore $\sigma_2 = w_2/w_2'$.

\item Subcase $\sigma_1\cdot\sigma_3 < \sigma_2\cdot\sigma_3$.
This case does not occur because $\sigma_1\cdot\sigma_3 = w_1/w_3'w_{13}$ and 
$\sigma_2\cdot\sigma_3 = w_2w_{23}/w_3'$.
\end{itemize}

\item[Case 3]
$\sigma_1\cdot\sigma_3 = \sigma_1'\cdot\sigma_3$
for strictly compatible $\sigma_1'$ and $\sigma_3$, and
$\sigma_2\cdot\sigma_3 = \sigma_2'\cdot\sigma_3$
for strictly compatible $\sigma_2'$ and $\sigma_3$.
From $(\sigma_1 \cdot \sigma_3) \sqcup (\sigma_2 \cdot \sigma_3) \ne \top$,
we have $\sigma_1'\cdot\sigma_3 = \sigma_2' \cdot\sigma_3$.
Then, $\sigma_1' = \sigma_2'$.
Hence, we have $\sigma_1 \le \sigma_2$ or $\sigma_2 \le \sigma_1$ 
by Lemma~\ref{lem:lele}.

\item[Case 4] $\sigma_1\cdot\sigma_3 = \sigma_1'\cdot\sigma_3$
for strictly compatible $\sigma_1'$ and $\sigma_3$, and
$\sigma_2\cdot\sigma_3 = \sigma_2\cdot\sigma_3'$
for strictly compatible $\sigma_2$ and $\sigma_3'$.
This case is the same as the case 2 by exchanging $\sigma_1$ and 
$\sigma_2$.\qed
\end{description}

\section{Correspondence for Well-Structured Pushdown Systems} \label{appendix:wellpds}

\begin{restatement}{Restatement of Proposition~\ref{prop:wellpds}} \mbox{}

\begin{itemize}
\item If $\utransrule{\tup{p_1,w_1}}{\tup{p_2,w_2}}{}{\mathcal{P}}$,
then $\utransrule{\tup{p_1,m_1}}{\tup{p_2,m_2}}{a}{\mathcal{P}'}$ and $w_1 \in  a(w_2)$.
\item If $\utransrule{\tup{p_1,m_1}}{\tup{p_2,m_2}}{a}{\mathcal{P}'}$ and 
$w_1 \in  a(w_2)$, then 
$\utransrule{\tup{p_1,w_1}}{\tup{p_2,w_2'}}{}{\mathcal{P}}$ for some $w_2 \preceq w_2'$.
\end{itemize}
 where $m_1 = |w_1|$ and $m_2 = |w_2|$.
\end{restatement}
\proof \mbox{}

\begin{itemize}
\item We prove the first statement by induction on the derivation of $\utransrule{\tup{p_1,w_1}}{\tup{p_2,w_2}}{}{\mathcal{P}}$.
\begin{description}
\item[Case] $\utransrule{\tup{p,w}}{\tup{p,w}}{}{}$ where $|w| = m$.
Then, 
$\utransrule{\tup{p,m}}{\tup{p,m}}{a}{\mathcal{P}'}$
where $a = \lambda w. \upcl{\set{w}})$.
Then, $w \in a(w)$.

\item[Case] $\utransrule{\tup{p,\gamma w'}}{\tup{p',\phi(\gamma)w'}}{}{}$,
$|w'| = m$, and $|\phi(\gamma)| = i$.
Then, 
$\utransrule{\tup{p_1,m+1}}{\tup{p_2,m+i}}{a}{\mathcal{P}'}$
where $a = \lambda \tup{w,w'}. \phi^{-1}(\upcl{\set{w}}) \times \upcl{\set{w'}}$.
Then, we have $a(\phi(\gamma)w') = \phi^{-1}(\upcl{\set{\phi(\gamma)}}) \times \upcl{\set{w'}} \ni \gamma w'$.

\item[Case] 
$\utransrule{\tup{p_1,w_1}}{\tup{p_3,w_3}}{}{\mathcal{P}}$ is obtained from
$\utransrule{\tup{p_1,w_1}}{\tup{p_2,w_2}}{}{\mathcal{P}}$ and
$\utransrule{\tup{p_2,w_2}}{\tup{p_3,w_3}}{}{\mathcal{P}}$.
By the induction hypotheses we have
$\utransrule{\tup{p_1,m_1}}{\tup{p_2,m_2}}{a_1}{\mathcal{P}'}$ and $w_1 \in  a(w_2)$,
as well as
$\utransrule{\tup{p_2,m_2}}{\tup{p_3,m_3}}{a_2}{\mathcal{P}'}$ and $w_2 \in  a(w_3)$.
Then, $\utransrule{\tup{p_1,m_1}}{\tup{p_2,m_2}}{a_1\odot a_2}{\mathcal{P}'}$
and $a_1\odot a_2(w_3) = \bigcup_{w \in a_2({w_3})} a_1(w) \supseteq
a_1(w_2) \ni w_1$.

\end{description}
\item We prove the second statement by induction on the derivation of 
$\utransrule{\tup{p_1,m_1}}{\tup{p_2,m_2}}{a}{\mathcal{P}'}$.

\begin{description}
\item[Case] $\utransrule{\tup{p,\#^{m}}}{\tup{p,\#^{m}}}{1_{m/m}}{\mathcal{P}'}$.
Let $w_1 \in \upcl{\set{w_2}} = 1_{m/m}(w_2)$.
Then, $\utransrule{\tup{p,w_1}}{\tup{p,w_1}}{}{\mathcal{P}}$ and $w_2 \preceq w_1$.

\item[Case]
  $\utransrule{\tup{p_1,\#^{m+1}}}{\tup{p_2,\#^{m+i}}}{a}{\mathcal{P}'}$
  is obtained from the fact that $(p_1, p_2, \phi) \in \Delta$ and from
$a = \lambda \tup{w,w'}. \phi^{-1}(\upcl{\set{w}}) \times \upcl{\set{w'}}$.
Let $w_2 = w_2'w_2''$ and $w_1 = \gamma w_1''$ where $|w_2'| = i$ and $|w_1''| = |w_2''| = m$.
Let $\gamma \in \phi^{-1}(\upcl{\set{w_2'}})$ and $w_1'' \in \upcl{\set{w_2''}}$.
Then, $\phi(\gamma) = w_2'''$ for some $w_2' \preceq w_2'''$.

Hence,
$\utransrule{\tup{p_1,\gamma w_1''}}{\tup{p_2,w_2''' w_1''}}{}{\mathcal{P}}$
and $w_2 = w_2'w_1'' \preceq w_2''' w_1''$.

\item[Case]
$\utransrule{\tup{p_1,m_1}}{\tup{p_3,m_3}}{a_1 \odot
  a_2}{\mathcal{P}'}$ is obtained from 
transitions
$\utransrule{\tup{p_1,m_1}}{\tup{p_2,m_2}}{a_1}{\mathcal{P}'}$ and 
$\utransrule{\tup{p_2,m_2}}{\tup{p_3,m_3}}{a_2}{\mathcal{P}'}$.
Let $w_1 \in  a_1 \odot a_2 (w_3) = \bigcup_{w \in a_2({w_3})} a_1(w)$.
Then, $w_1 \in a_1(w_2)$ and $w_2 \in a_2(w_3)$ for some $w_2$.
By the induction hypothesis,
$\utransrule{\tup{p_1,w_1}}{\tup{p_2,w_2'}}{}{\mathcal{P}}$ for some $w_2 \preceq w_2'$ and
$\utransrule{\tup{p_2,w_2}}{\tup{p_2,w_3'}}{}{\mathcal{P}}$ for some $w_3 \preceq w_3'$.
By the monotonicity of $\mathcal{P}$, 
$\utransrule{\tup{p_2,w_2'}}{\tup{p_2,w_3''}}{}{\mathcal{P}}$ for some $w_3' \preceq w_3''$.
Then, $\utransrule{\tup{p_1,w_1}}{\tup{p_3,w_3''}}{}{\mathcal{P}}$ and
$w_3 \preceq w_3''$.\qed

\end{description}
\end{itemize}

\section{Correspondence for Pushdown Systems with Stack Manipulation} \label{appendix:trpds}

\begin{restatement}{Restatement of Proposition~\ref{prop:trpds}} 
Let $m_1 = |w_1|$ and $m_2 = |w_2|$.
\begin{itemize}
\item If $\utransrule{\tup{p_1,w_1}}{\tup{p_2,w_2}}{}{\mathcal{P}}$,
then $\utransrule{\tup{p_1,m_1}}{\tup{p_2,m_2}}{a}{\mathcal{P}'}$ and 
$\tup{\epsilon,\epsilon} \in a(w_1,w_2)$ for some $a$.

\item If $\utransrule{\tup{p_1,m_1}}{\tup{p_2,m_2}}{a}{\mathcal{P}'}$ and 
$\tup{\epsilon,\epsilon} \in a(w_1,w_2)$, then 
$\utransrule{\tup{p_1,w_1}}{\tup{p_2,w_2}}{}{\mathcal{P}}$.
\end{itemize}
\end{restatement}
\proof Let $|w_i| = m_i$ for $1 \le i \le 3$ in this proof.
\begin{itemize}
\item We prove the first statement 
by induction on the derivation of
$\utransrule{\tup{p_1,w_1}}{\tup{p_2,w_2}}{}{\mathcal{P}}$.

\begin{description}
\item[Case] $\utransrule{\tup{p_1,w_1}}{\tup{p_1,w_1}}{}{\mathcal{P}}$.
We have $\utransrule{\tup{p,m_1}}{\tup{p,m_1}}{1_{m_1/m_1}}{\mathcal{P}'}$
and $1_{m_1/m_1}(w_1,w_1) =  1_{\mathcal{T}} \ni \tup{\epsilon,\epsilon}$.

\item[Case] $\utransrule{\tup{p,\gamma w'}}{\tup{p',ww''}}{}{\mathcal{P}}$ is obtained
from $\tup{p, \gamma, p', w, t} \in \Delta$ and $\tup{w',w''} \in t$.
Let $|w| = n$ and $|w'| = |w''| = m$.
Then, 
$\utransrule{\tup{p_1,m+1}}{\tup{p_2,m+n}}{a}{\mathcal{P}'}$
where $a = \ext{1/n,1+m/n+m}{\trstep{t}{\gamma}{w}}$ and
$a(\gamma w',w w'') = \tup{w',w''}^{-1}(\trstep{t}{\gamma}{w}(\gamma,w))
= \tup{w',w''}^{-1}t \ni \tup{\epsilon,\epsilon}$.

\item[Case] $\utransrule{\tup{p_1,w_1}}{\tup{p_3,w_3}}{}{\mathcal{P}}$ is obtained from
$\utransrule{\tup{p_1,w_1}}{\tup{p_2,w_2}}{}{\mathcal{P}}$ and
$\utransrule{\tup{p_2,w_2}}{\tup{p_3,w_3}}{}{\mathcal{P}}$.
By the induction hypotheses, 
$\utransrule{\tup{p_1,m_1}}{\tup{p_2,m_2}}{a_1}{\mathcal{P}'}$,
$\utransrule{\tup{p_2,m_2}}{\tup{p_3,m_3}}{a_2}{\mathcal{P}'}$,
$\tup{\epsilon,\epsilon} \in a_1(w_1,w_2)$, and
$\tup{\epsilon,\epsilon} \in a_2(w_2,w_3)$.
Then,  $\utransrule{\tup{p_1,m_1}}{\tup{p_3,m_3}}{a_1 \odot a_2}{\mathcal{P}'}$
and $\tup{\epsilon,\epsilon} \in a_1(w_1,w_2) \circ a_2(w_2,w_3) \subseteq
a_1 \odot a_2 (w_1,w_3)$.
\end{description}

\item We prove the second statement 
by induction on the derivation of
$\utransrule{\tup{p_1,m_1}}{\tup{p_2,m_2}}{a}{\mathcal{P}'}$.

\begin{description}
  
\item[Case]$\utransrule{\tup{p,m}}{\tup{p,m}}{1_{m/m}}{\mathcal{P}'}$ and
$\tup{\epsilon,\epsilon} \in 1_{m/m}(w_1,w_2)$. 
By the definition of $1_{m/m}$, $w_1 = w_2$.
Thus, $\utransrule{\tup{p,w_1}}{\tup{p,w_2}}{}{\mathcal{P}}$.

\item[Case] $\utransrule{\tup{p_1,m+1}}{\tup{p_2,m+n}}{a}{\mathcal{P}'}$
where $a = \ext{1/n,1+m/n+m}{\trstep{t}{\gamma}{w}}$. 
Let $\tup{\epsilon,\epsilon} \in a(\gamma_0 w',w_0 w'')$ where
$|w'| = |w''| = m$, $|w_0| =n$. 
\begin{eqnarray*}
a(\gamma_0 w',w_0 w'') & = & \ext{1/n,1+m/n+m}{\trstep{t}{\gamma}{w}} (\gamma_0 w',w_0 w'') \\
& = & \tup{w',w''}^{-1} (\trstep{t}{\gamma}{w}(\gamma_0,w_0)) 
\end{eqnarray*}
Then, we have $\gamma_0 = \gamma$, $w_0 = w$, and $\tup{\epsilon,\epsilon} \in \tup{w',w''}^{-1}t$, i.e., $\tup{w',w''} \in t$. 
Hence, $\utransrule{\tup{p,\gamma w'}}{\tup{p',ww''}}{}{}$.

\item[Case] 
$\utransrule{\tup{p_1,m_1}}{\tup{p_3,m_3}}{a_1 \odot a_2}{\mathcal{P}'}$ is
obtained from
$\utransrule{\tup{p_1,m_1}}{\tup{p_2,m_2}}{a_1}{\mathcal{P}'}$ and $\utransrule{\tup{p_2,m_2}\\}{\tup{p_3,m_3}}{a_2}{\mathcal{P}'}$.
Let $\tup{\epsilon,\epsilon} \in a_1 \odot a_2 (w_1,w_3)$.
Then, $\tup{\epsilon,\epsilon} \in a_1 (w_1,w_2) \circ a_2(w_2,w_3)$
for some $w_2$. Since $a_1(w_1,w_2)$ and $a_2(w_2,w_3)$ are letter-to-letter transducers,
$\tup{\epsilon,\epsilon} \in a_1 (w_1,w_2)$
and $\tup{\epsilon,\epsilon} \in a_2 (w_2,w_3)$.
Then, we obtain 
$\utransrule{\tup{p_1,w_1}}{\tup{p_3,w_3}}{}{\mathcal{P}}$ from the induction hypotheses.\qed
\end{description}
\end{itemize}

\end{document}